\documentclass[11pt]{amsart}
\usepackage{amssymb,amsmath,amsfonts,stmaryrd}
\usepackage[utf8]{inputenc}
\usepackage{microtype}
\usepackage[usenames, dvipsnames]{xcolor}
\usepackage[backref=page, colorlinks=true]{hyperref}
\hypersetup{
 linkcolor={red!50!black},
 citecolor={green!50!black},
 urlcolor={blue!80!black}
}
\usepackage{tikz-cd}
\usepackage[margin=1.3in]{geometry}
\usepackage[capitalize, noabbrev]{cleveref} 
\numberwithin{equation}{section}
\usepackage{amsthm}
\newtheorem{thm}[equation]{Theorem}
\newtheorem{lem}[equation]{Lemma}
\newtheorem{prop}[equation]{Proposition}
\newtheorem{cor}[equation]{Corollary}
\theoremstyle{definition}
\newtheorem{defn}[equation]{Definition}
\theoremstyle{remark}
\newtheorem{rem}[equation]{Remark}
\newtheorem{question}[equation]{Question}

\crefname{thm}{Theorem}{Theorems}
\usepackage{mathtools}
\DeclarePairedDelimiter{\set}{\{}{\}}
\usepackage{adamsmacros}
\usepackage{tikz-cd}
\usepackage{subcaption}
\usepackage{setspace}
\usetikzlibrary{calc}

\usepackage[mathscr]{eucal}
\makeatletter
\def\instring#1#2{TT\fi\begingroup
  \edef\x{\endgroup\noexpand\in@{#1}{#2}}\x\ifin@}
\def\isuppercase#1{%
  \instring{#1}{ABCDEFGHIJKLMNOPQRSTUVWXYZ}%
}%
\newcommand{\C@lIfUpper}[1]{
 \if\isuppercase{#1}\mathscr{#1}%
 \else #1%
 \fi
}
\newcommand{\cat}[1]{\mathit{\@tfor\next:=#1\do{\C@lIfUpper{\next}}}}
\makeatother

\usepackage{comment}
\usepackage{xparse}
\newcommand{\Z}{\mathbb Z}
\newcommand{\bZ}{\mathbb Z}

\newcommand{\C}{\mathbb C}
\newcommand{\U}{\mathrm U}
\newcommand{\Spin}{\mathrm{Spin}}

\newcommand{\Sq}{\mathrm{Sq}}
\newcommand{\abs}[1]{\lvert #1 \rvert}
\newcommand{\pt}{\mathrm{pt}}

\newcommand{\SU}{\mathrm{SU}}

\newcommand{\SO}{\mathrm{SO}}
\newcommand{\SL}{\mathrm{SL}}

\newcommand{\Mp}{\mathrm{Mp}}
\renewcommand{\O}{\mathrm O}

\newcommand{\RP}{\mathbb{RP}}
\newcommand{\bl}{\text{--}}
\newcommand{\R}{\mathbb R}

\newcommand{\su}{\mathfrak{su}}
\newcommand{\cN}{\mathcal N}
\newcommand{\Q}{\mathbb Q}

\newcommand{\term}{\emph}

\newcommand{\cA}{\mathcal A}
\newcommand{\Ext}{\mathrm{Ext}}
\newcommand{\Hom}{\mathrm{Hom}}

\newcommand{\vp}{\varphi}
\newcommand{\MTH}{\mathit{MTH}}
\newcommand{\MTSO}{\mathit{MTSO}}

\newcommand\MAILTO[1]{\href{mailto:#1}{\nolinkurl{#1}}}

\DeclareDocumentCommand{\shortexact}{s O{} O{} mmmm}{
\IfBooleanTF{#1}{ 
\begin{tikzcd}[ampersand replacement=\&]
	{1} \& {#4} \& {#5} \& {#6} \& {1#7}
	\arrow[from=1-1, to=1-2]
	\arrow["#2", from=1-2, to=1-3]
	\arrow["#3", from=1-3, to=1-4]
	\arrow[from=1-4, to=1-5]
\end{tikzcd}
}{ 
\begin{tikzcd}[ampersand replacement=\&]
	{0} \& {#4} \& {#5} \& {#6} \& {0#7}
	\arrow[from=1-1, to=1-2]
	\arrow["#2", from=1-2, to=1-3]
	\arrow["#3", from=1-3, to=1-4]
	\arrow[from=1-4, to=1-5]
\end{tikzcd}
}}

\newcommand{\Mzero}[3]{
	\tikzpt{#1}{#2}{#3}{};
	\foreach \y in {2, 3, 5} {
		\tikzpt{#1}{#2 + \y}{}{};
	}
	\sqone(#1, #2 + 2);
	\sqtwoL(#1, #2);
	\sqtwoL(#1, #2 + 3);
}

\title{What bordism-theoretic anomaly cancellation can do for U}
\date{\today}

\author{Arun Debray}
\address{Purdue University, West Lafayette, Indiana}
\email{adebray@purdue.edu}

\author{Matthew Yu}
\address{Perimeter Institute for Theoretical Physics, Waterloo, Ontario}
\email{myu@perimeterinstitute.ca}
\thanks{It is a pleasure to thank Diego Delmastro, Markus Dierigl, Dan Freed, Jonathan J. Heckman, Theo Johnson-Freyd, Miguel Montero, and David Speyer for helpful comments. We also thank the Simons Collaboration for Global Categorical Symmetries for hosting a summer school where some of this work was conducted. Research at the Perimeter Institute is supported by the Government of Canada through Industry Canada and by the Province of Ontario through the Ministry of Economic Development and Innovation. 
}

\begin{document}
\maketitle

\begin{abstract}
We perform a bordism computation to show that the $E_{7(7)}(\R)$ U-duality symmetry of 4d $\mathcal N = 8$
supergravity could have an anomaly invisible to perturbative methods; then we show that this anomaly is trivial. We
compute the relevant bordism group using the Adams and Atiyah-Hirzebruch spectral sequences, and we show the
anomaly vanishes by computing $\eta$-invariants on the Wu manifold, which generates the bordism group.

\end{abstract}

\tableofcontents

\section{Introduction}

One of the most surprising discoveries in the field of string theory is the existence of duality symmetries. These symmetries show that the same theory can be described in superficially different ways. In some cases, this can be seen via a transformation of the parameters of the theory, or even the spacetime itself. One such symmetry is U-duality, given by the group $E_{n(n)}(\bZ)$. By starting with an 11-dimensional theory which encompasses the type IIA string theory, and compactifying on an $n$-torus, we gain an $\mathrm{SL}_n(\bZ)$ symmetry from the mapping class group on the $n$-torus. We arrive at the same theory by compactifying 10d type IIB on an $n-1$-torus, and obtain an $\O_{n-1,n-1}(\bZ)$ symmetry related to T-duality. The group $E_{n(n)}(\bZ)$ is then generated by the two aforementioned groups.

In the low energy regime of the 11d theory, which is 11d supergravity, we have an embedding of  $E_{n(n)}(\bZ)\hookrightarrow E_{n(n)}$ upon applying the torus compactification procedure. The latter group is the U-duality of supergravity. One finds a
maximally noncompact form of $E_n$ after the compactification, and this is denoted $E_{n(n)}(\mathbb{R})$. The maximally noncompact
form of a Lie group of rank $n$ contains $n$ more noncompact generators than
compact generators. 
For the purpose of this paper, we reduce 11-dimensional supergravity on a 7-dimensional torus. This gives a maximal supergravity theory, i.e. 4d $\mathcal{N}=8$ supergravity, with an $E_7$ symmetry.\footnote{
Dimensional reduction of IIB supergravity on an 6-dimensional torus also yields the same symmetry.} The noncompact form is $E_{7(7)}$ which is $133$-dimensional and is diffeomorphic, but not isomorphic, to	$\SU_8/\set{\pm 1}\times\R^{70}$.

Because this is a symmetry of the theory, one can ask if it is anomalous, and in particular if there are any global anomalies. Since 4d $\mathcal{N}=8$ supergravity arises as the low energy effective theory of string theory, then a strong theorem of quantum gravity saying that there are no global symmetries implies that the U-duality symmetry must be gaugeable.
Therefore,
the existence of any global anomaly would require a mechanism for its cancellation. It would therefore be an interesting question to consider if additional topological terms  need to be added to cancel the nonperturbative anomaly as in \cite{DDHM21}, but 
we show that with the matter content of 4d maximal supergravity is sufficient to cancel the anomaly on the nose.

The purpose of this paper is to answer: 
\begin{question}
Can 4d $\mathcal{N}=8$ supergravity with an $E_{7(7)}$ symmetry have a nontrivial anomaly topological field theory (TFT)? If it can, how do we show that the anomaly vanishes?
\end{question}
We find that theories with this symmetry type can have a nontrivial anomaly, so we have to check whether 4d $\cN = 8$ supergravity carries this nontrivial anomaly.
\begin{thm}
The group of deformation classes of 5d reflection-positive, invertible TFTs on spin-$\SU_8$ manifolds is isomorphic to $\Z/2$. In this group, the anomaly field theory of 4d $\cN = 8$ supergravity is trivial.
\end{thm}
\begin{rem}
We also show that the nontrivial deformation class contains the invertible field theory whose partition function on
a closed $5$-manifold $M$ is $\int_m w_2(M) w_3(M)$, which is nonzero on the Wu manifold. This theory, and its
appearance as the anomaly theory of some 4d QFTs, are discussed in detail by Wang-Wen-Witten~\cite{WWW19}, who
describe how this anomaly can be canceled by tensoring with a TFT carrying the same anomaly.
\end{rem}

%

The order of the global anomaly is equal to the order of a bordism group in degree 5 that can be computed from the
Adams spectral sequence. We find that the global anomaly is $\bZ/2$ valued, but nonetheless is trivial when we take
into account the matter content of 4d $\mathcal{N}=8$ supergravity. In order to see the cancellation we first find
the manifold generator of the bordism group, which happens to be the Wu manifold, and compute $\eta$-invariants on
it.
This bordism computation is also mathematically intriguing because we find ourselves working over the entire Steenrod algebra, however the specific properties of the problem we are interested in make this tractable.

This work only focuses on U-duality as the group $E_{7(7)}$ rather than $E_{7(7)}(\bZ)$, because the cohomology of the discrete group that arises in string theory is not known, and a strategy we employ of taking the maximal compact subgroup will not work. But one could imagine running a similar Adams computation for the group $E_7(\bZ)$ and checking that the anomaly vanishes. 
There are also a plethora of dualities that arise from compactifying 11d supergravity that one can also compute anomalies of, among them are the U-dualities that arise from compactifying on lower dimensional tori. 

The structure of the paper is as follows: in \S\ref{sec:BordismSetup} we present the symmetries and tangential structure for the maximal 4d supergravity theory with U-duality symmetry and turn it into a bordism computation. We also give details on
the field content of the theory and how it is compatible with the type of manifold we are considering. 
In \S\ref{section:InvTheory} we review the possibility of global anomalies, and invertible field theories.
In \S\ref{sec:SSComputation} we perform the spectral sequence computation and give the manifold generator for the bordism group in question. In \S\ref{sec:EvalAnomaly} we show that the anomaly vanishes by considering the field content on the manifold generator.


\section{Placing the U-duality symmetry on manifolds}\label{sec:BordismSetup}
In this section, we review how the $E_{7(7)}$ U-duality symmetry acts on the fields of 4d $\cN = 8$
supergravity; then we discuss what kinds of manifolds are valid backgrounds in the presence of this symmetry. 
We assume that we have already Wick-rotated into Euclidean signature.  We determine a Lie group $H_4$ with a map
$\rho_4\colon H_4\to \O_4$ such that 4d $\cN = 8$ supergravity can be formulated on $4$-manifolds $M$ equipped with
a metric and an $H_4$-connection $P,\Theta\to M$, such that $\rho_4(\Theta)$ is the Levi-Civita connection. As we
review in \S\ref{section:InvTheory}, anomalies are classified in terms of bordism; once we know $H_4$ and $\rho_4$,
Freed-Hopkins' work~\cite{FH21a} tells us what bordism groups to compute.

The field content of 4d $\cN = 8$ supergravity coincides with the spectrum of type IIB closed string theory compactified on $T^6$ and consists of the following fields:
\begin{itemize}
    \item $70$ scalar fields,
    \item $56$ gauginos (spin $1/2$),
    \item $28$ vector bosons (spin $1$),
    \item $8$ gravitinos (spin $3/2$), and
    \item $1$ graviton (spin $2$).
\end{itemize}
Cremmer-Julia~\cite{CJ79} exhibited an $\mathfrak e_{7(7)}$ symmetry of this theory, meaning an action on the
fields for which the Lagrangian is invariant. Here, $\mathfrak e_{7(7)}$ is the Lie algebra of the real, noncompact
Lie group $E_{7(7)}$, which is the split form of the complex Lie group $E_7(\C)$. Cartan~\cite[\S VIII]{Car14} constructed $E_{7(7)}$
explicitly as follows: the $56$-dimensional vector space
\begin{equation}
	V \coloneqq \Lambda^2(\R^8)\oplus \Lambda^2 ((\R^8)^*)
\end{equation}
has a canonical symplectic form coming from the duality pairing. $E_{7(7)}$ is defined to be the subgroup of
$\mathrm{Sp}(V)$ preserving the quartic form
\begin{equation}
	q(x^{ab}, y_{cd}) = x^{ad}y_{bc}x^{cd}y_{da} - \frac 14 x^{ab}y_{ab}x^{cd}y_{cd} + \frac{1}{96}\left(
		\epsilon_{abc\dotsm h}x^{ab} x^{cd} x^{ef} x^{gh} + \epsilon^{abc\dotsm h}
		y_{ab}y_{cd}y_{ef}y_{gh}\right).
\end{equation}
Thus, by construction, $E_{7(7)}$ comes with a $56$-dimensional representation, which we denote $\boldsymbol{56}$.

$E_{7(7)}$ is noncompact; its maximal compact is $\SU_8/\set{\pm 1}$, giving us an embedding
$\su_8\subset \mathfrak e_{7(7)}$. Thus $\pi_1(E_{7(7)})\cong\Z/2$; let $\widetilde E_{7(7)}$ denote the
universal cover, which is a double cover.

There is an action of $\mathfrak e_{7(7)}$ on the fields of 4d $\cN = 8$ supergravity, but in this paper we only need to know how $\su_8\subset\mathfrak e_{7(7)}$ acts: we will see in \S\ref{section:Usymtype} that the anomaly calculation factors through the maximal compact subgroup of $E_{7(7)}$. For the full $\mathfrak e_{7(7)}$ story, see~\cite[\S 2]{FM13}; the $\mathfrak e_{7(7)}$-action exponentiates to an $\widetilde
E_{7(7)}$-action on the fields. The $\su_8$-action is:
%
%
\begin{enumerate}
	\item The $70$ scalar fields can be repackaged into a single field valued in $E_{7(7)}/(\SU_8/\set{\pm 1})$ with trivial $\su_8$-action.
    \item The gauginos transform in the representation $\mathbf{56}\coloneqq \Lambda^3(\C^8)$.
    \item The vector bosons transform in the $28$-dimensional representation $\Lambda^2(\C^8)$, which we call
	$\boldsymbol{28}$.
    \item The gravitinos transform in the defining representation of $\su_8$, which we denote
	$\boldsymbol 8$.
    \item The graviton transforms in the trivial representation.
\end{enumerate}

The presence of fermions (the gauginos and gravitinos) means that we must have data of a spin structure, or
something like it, to formulate this theory. In quantum physics, a strong form of $G$-symmetry is to couple to a
$G$-connection, suggesting that we should formulate 4d $\cN = 8$ supergravity on Riemannian spin $4$-manifolds $M$ together
with an $\widetilde E_{7(7)}$-bundle $P\to M$ and a connection on $P$. The spin of each field tells us which
representation of $\Spin_4$ it transforms as, and we just learned how the fields transform under the $\widetilde
E_{7(7)}$-symmetry, so we can place this theory on manifolds $M$ with a \term{geometric $\Spin_4\times\widetilde
E_{7(7)}$-structure}, i.e.\ a metric and a principal $\Spin_4\times\widetilde E_{7(7)}$-bundle $P\to M$ with
connection whose induced $\O_4$-connection is the Levi-Civita connection. The fields are sections of associated
bundles to $P$ and the representations they transform in. The Lagrangian is invariant under the
$\Spin_4\times\widetilde E_{7(7)}$-symmetry, so defines a functional on the space of fields, and we can study this
field theory as usual.

However, we can do better! We will see that the representations above factor through a quotient $H_4$ of
$\Spin_4\times\widetilde E_{7(7)}$, which we define below in~\eqref{H4defn}, so the same procedure above works with $H_4$ in place of
$\Spin_4\times\widetilde E_{7(7)}$. A lift of the structure group to $H_4$ is less data than a lift all the way to
$\Spin_4\times\widetilde E_{7(7)}$, so we expect to be able to define 4d $\cN = 8$ supergravity on more manifolds.
This is similar to the way that the $\SL_2(\Z)$ duality symmetry in type IIB string theory can be placed not just
on manifolds with a $\Spin_{10}\times\Mp_2(\Z)$-structure,\footnote{Here $\Mp_2(\Z)$ is the \term{metaplectic group}, a central extension of $\SL_2(\Z)$ of the form\label{metaplectic_footnote}
\begin{equation}
	\shortexact*{\set{\pm 1}}{\Mp_2(\Z)}{\SL_2(\Z)}.
\end{equation}} but on the larger class of manifolds with a
$\Spin_{10}\times_{\set{\pm 1}}\Mp_2(\Z)$-structure~\cite[\S 5]{Pantev:2016nze}, or how certain $\SU_2$ gauge
theories can be defined on manifolds with a $\Spin_n\times_{\set{\pm 1}}\SU_2$ structure~\cite{WWW19,albanese2021spinh,lawson2023spin}.

Let $-1\in\Spin_4$ be the nonidentity element of the kernel of $\Spin_4\to\SO_4$ and let $x$ be the nonidentity
element of the kernel of $\widetilde E_{7(7)}\to E_{7(7)}$. The key fact allowing us to descend to a quotient is
that $-1$ acts nontrivially on the representations of $\Spin_4\times\widetilde E_{7(7)}$ above, and $x$ acts
nontrivially, but on a given representation, these two elements both act by $1$ or they both act by $-1$. We can check this even though we have not specified the entire $\mathfrak e_{7(7)}$-action on the fields, because $-1\in\widetilde E_{7(7)}$ is contained in the copy of $\SU_8$ in $\widetilde E_{7(7)}$, and we have specified the $\su_8$-action. Therefore the $\Z/2$
subgroup of $\Spin_4\times\widetilde E_{7(7)}$ generated by $(-1, x)$ acts trivially, and we can form the quotient
\begin{equation}
\label{H4defn}
	H_4\coloneqq \Spin_4\times_{\set{\pm 1}} \widetilde E_{7(7)} = (\Spin_4\times\widetilde E_{7(7)})/\langle (-1,
	x)\rangle.
\end{equation}
The representations that the fields transform in all descend to representations of $H_4$, so following the
procedure above, we can define 4d $\cN = 8$ supergravity on manifolds $M$ with a \term{geometric $H_4$-structure},
i.e.\ a metric, an $H_4$-bundle $P\to M$, and a connection on $P$ whose induced $\O_4$-connection is the
Levi-Civita connection.

\begin{rem}
As a check to determine that we have the correct symmetry group, we can compare with other string dualities. The
U-duality group contains the S-duality group for type IIB string theory, which comes geometrically from the fact
that 4d $\cN = 8$ supergravity can be constructed by compactifying type IIB string theory on $T^6$. Therefore the
ways in which the duality groups mix with the spin structure must be compatible. As explained by
Pantev-Sharpe~\cite[\S 5]{Pantev:2016nze}, the $\SL_2(\Z)$ duality symmetry of type IIB string theory mixes with
the spin structure to form the group $\Spin_{10}\times_{\set{\pm 1}}\Mp_2(\Z)$, where $\Mp_2(\Z)$ is the metaplectic group from Footnote~\ref{metaplectic_footnote}.

Therefore the way in which the U-duality group mixes with $\set{\pm 1}\subset \Spin_4$ must also be nontrivial.
Extensions of a group $G$ by $\set{\pm 1}$ are classified by $H^2(BG;\set{\pm 1})$. If $G$ is connected, $BG$ is
simply connected and the Hurewicz and universal coefficient theorems together provide a natural identification
\begin{equation}
	H^2(BG;\set{\pm 1})\overset\cong\longrightarrow \Hom(\pi_2(BG), \set{\pm 1}) = \Hom(\pi_1(G), \set{\pm 1}).
\end{equation}
As $\pi_1(E_{7(7)})\cong\Z/2$, there is only one nontrivial extension of $E_{7(7)}$ by $\set{\pm 1}$, namely the
universal cover $\widetilde E_{7(7)}\to E_{7(7)}$. That is, by comparing with S-duality, we again obtain the group
$H_4$, providing a useful double-check on our calculation above.
\end{rem}

\section{Anomalies, invertible field theories, and bordism}\label{section:InvTheory}
\subsection{Generalities on anomalies and invertible field theories}
It is sometimes said that in mathematical physics, if you ask four people what an anomaly is, you will get five
answers. The goal of this section is to explain our perspective on anomalies, due to Freed-Teleman~\cite{FT14}, and how to
reduce the determination of the anomaly to a question in algebraic topology, an approach due to
Freed-Hopkins-Teleman~\cite{FHT10} and Freed-Hopkins~\cite{FH21a}.

Whatever an anomaly is, it signals a mild inconsistency in the definition of a quantum field theory. For example,
if a quantum field theory $Z$ is $n$-dimensional, one ought to be able to evaluate it on a closed $n$-manifold $M$,
possibly equipped with some geometric structure, to obtain a complex number $Z(M)$, called the \term{partition
function} of $M$. If $Z$ has an anomaly, $Z(M)$ might only be defined after some additional choices, and in the
absence of those choices $Z(M)$ is merely an element of a one-dimensional complex vector space $\alpha(M)$.

The theory $Z$ is local in $M$, so $\alpha(M)$ should also be local in $M$. One way to express this locality is to
ask that $\alpha(M)$ is the state space of $M$ for some $(n+1)$-dimensional quantum field theory $\alpha$, called
the \term{anomaly field theory} $\alpha$ of $Z$. The condition that the state spaces of $\alpha$ are
one-dimensional follows from the fact that $\alpha$ is an \term{invertible field theory}~\cite[Definition 5.7]{FM06}, meaning that there is some other field theory $\alpha^{-1}$ such that $\alpha\otimes\alpha^{-1}$ is isomorphic to
the trivial field theory $\boldsymbol 1$.\footnote{The relationship between invertibility and one-dimensional state
spaces is that $\alpha\otimes\alpha^{-1} \simeq\boldsymbol 1$ means that on any closed, $n$-manifold $M$, there is
an isomorphism of complex vector spaces $\alpha(M)\otimes\alpha^{-1}(M)\cong \boldsymbol 1(M) = \C$. This forces
$\alpha(M)$ and $\alpha^{-1}(M)$ to be one-dimensional. Often the converse is also true: see
Schommer-Pries~\cite{Sch18}.}\textsuperscript{,}\footnote{In some cases, we do not want to assume $\alpha$ extends to closed $n$-manifolds; see Freed-Teleman~\cite{FT14} for more information. But the U-duality anomaly we investigate in this paper does extend.}
This approach to anomalies is due to Freed-Teleman~\cite{FT14}; see also Freed~\cite{Freed:2014iua, Fre19}.
%

We can therefore understand the possible anomalies associated to a given $n$-dimensional quantum field theory $Z$
by classifying the $(n+1)$-dimensional invertible field theories with the same symmetry type as $Z$. The
classification of invertible \emph{topological} field theories is due to Freed-Hopkins-Teleman~\cite{FHT10}, who
lift the question into stable homotopy theory; see Grady-Pavlov~\cite[\S 5]{GP21} for a recent generalization to
the nontopological setting.

Supergravity with its U-duality symmetry is a unitary quantum field theory, and therefore its anomaly theory satisfies the Wick-rotated analogue of unitarity: \term{reflection positivity}.
Freed-Hopkins~\cite{FH21a} classify reflection-positive invertible field theories, again using stable homotopy
theory. Let $\O\coloneqq\lim_{n\to\infty}\O_n$ be the infinite orthogonal group.
\begin{thm}[{Freed-Hopkins~\cite[Theorem 2.19]{FH21a}}]
\label{stabilization}
Let $n\ge 3$, $H_n$ be a compact Lie group, and $\rho_n\colon H_n\to\O_n$ be a homomorphism whose image contains
$\SO_n$. Then there is canonical data of a topological group $H$ and a continuous homomorphism $\rho\colon H\to\O$
such that the pullback of $\rho$ along $\O_n\hookrightarrow\O$ is $\rho_n$.
\end{thm}
In other words, when the hypotheses of this theorem hold, we have more than just $H_n$-structures on $n$-manifolds;
we can define $H$-structures on manifolds of any dimension, by asking for a lift of the classifying map of the
stable tangent bundle $M\to B\O$ to $BH$; a manifold equipped with such a lift is called an \term{$H$-manifold}.
Following Lashof~\cite{Las63}, this allows us to define bordism groups $\Omega_k^H$ and a homotopy-theoretic object
called the \term{Thom spectrum} $\MTH$, whose homotopy groups are the $H$-bordism groups.\footnote{\label{MT_vs_M}In homotopy theory, it is common to work with $H$-structures on the stable normal bundle instead of on the tangent bundle; the Thom spectrum whose homotopy groups are bordism groups of manifolds with $H$-structures on their stable normal bundle is called $\mathit{MH}$. For all groups $H$ we discuss in this paper, an $H$-structure on the tangent bundle of $M$ induces an $H$-structure on the stable normal bundle of $M$ and vice versa, and $\MTH\simeq\mathit{MH}$; thus this distinction will not matter for us. However, in some other situations, it does matter; for example, $\mathit{MTPin}^+\not\simeq\mathit{MPin}^+$.} See~\cite[\S 2]{BC18} for more on the definition of $\MTH$ and its context in stable homotopy theory.
\begin{thm}[Freed-Hopkins~\cite{FH21a}]
\label{FHclassification}
With $H_n$ as in \cref{stabilization}, the abelian group of deformation classes of $n$-dimensional
reflection-positive invertible topological field theories on $H_n$-manifolds is naturally isomorphic to the torsion
subgroup of $[\MTH, \Sigma^{n+1}I_\Z]$.
\end{thm}
Freed-Hopkins then conjecture (\textit{ibid.}, Conjecture 8.37) that the whole group $[\MTH, \Sigma^{n+1}I_\Z]$
classifies all reflection-positive invertible field theories, topological or not.

The notation $[\MTH, \Sigma^{n+1}I_\Z]$ means the abelian group of homotopy classes of maps between $\MTH$ and an
object $\Sigma^{n+1}I_\Z$ belonging to stable homotopy theory; see~\cite[\S 6.1]{FH21a} for a brief introduction in
a mathematical physics context. We mentioned $\MTH$ above; $I_\Z$ is the \term{Anderson dual of the sphere
spectrum}~\cite{And69, Yos75}, characterized up to homotopy equivalence by its universal property, which says that
there is a natural short exact sequence
\begin{equation}
    \shortexact{\Ext(\pi_{n-1}(E), \Z)}{[E, \Sigma^nI_\Z]}{\Hom(\pi_n(E), \Z)}.
\end{equation}
Applying this when $E = \MTH$, we obtain a short exact sequence
\begin{equation}\label{eq:SESanomaly}
	\shortexact[\vp][\psi]{\Ext(\Omega_{n+1}^H, \Z)}{[\MTH, \Sigma^{n+2}I_\Z]}{\Hom(\Omega_{n+2}^H,
	\Z)}{}
\end{equation}
decomposing the group of possible anomalies of unitary QFTs on $H_n$-manifolds. These two factors admit
interpretations in terms of anomalies.
\begin{enumerate}
	\item The quotient $\Hom(\Omega_{n+2}^H, \Z)$ is a free abelian group of degree-$(n+2)$ characteristic
	classes of $H$-manifolds. The map $\psi$ sends an anomaly field theory to its anomaly polynomial. This is the
	part of the anomaly visible to perturbative methods, and sometimes is called the \term{local anomaly}.
	\item The subgroup $\Ext(\Omega_{n+1}^H, \Z)$ is isomorphic to the abelian group of torsion
	bordism invariants $f\colon \Omega_{n+1}^H\to\C^\times$. These classify the reflection-positive invertible
	\emph{topological} field theories $\alpha_f$: the correspondence is that the bordism invariant $f$ is the
	partition function of $\alpha_f$. This part of an anomaly field theory is generally invisible to perturbative
	methods and is called the \term{global anomaly}.
\end{enumerate}
Work of Yamashita-Yonekura~\cite{Yamashita:2021cao} and Yamashita~\cite{Yam21} relates this short exact sequence to
a differential generalized cohomology theory extending $\mathrm{Map}(\MTH, \Sigma^{n+1}I_\Z)$.
\subsection{Specializing to the U-duality symmetry type}\label{section:Usymtype}
For us, $n = 4$ and the symmetry type is $H_4 = \Spin\times_{\set{\pm 1}}\widetilde E_{7(7)}$. This group is not
compact, so \cref{stabilization,FHclassification} above do not apply. However, we can work around this obstacle:
Marcus~\cite{Mar85} proved that the anomaly polynomial of the $E_{7(7)}$ symmetry vanishes,\footnote{Marcus'
analysis does not discuss the question of $H_4$ versus $\Spin_4\times\widetilde E_{7(7)}$, but this does not
matter: in many cases including the one we study, the anomaly polynomial for a $d$-dimensional field theory on
$G$-manifolds is an element of $H^{d+2}(BG;\Q)$, and rational cohomology is insensitive to finite covers such as
$\Spin_4\times\widetilde E_{7(7)}\to H_4$. Thus Marcus' computation applies in our case too.} meaning that the
anomaly field theory is a \emph{topological} field theory. Thinking of topological field theories as symmetric
monoidal functors $\cat{Bord}_n^{H_n}\to\cat C$, we can freely adjust the structure we put on manifolds in these
theories as long as the induced map on bordism categories is an equivalence. We make two adjustments.
\begin{enumerate}
	\item First, forget the metric and connection in the definition of a geometric $H_4$-structure. The space of
	such data is contractible and therefore can be ignored for topological field theories.
	\item We can then replace $H_4$ with its maximal compact subgroup: for any Lie group $G$ with $\pi_0(G)$
	finite, inclusion of the maximal compact subgroup $K\hookrightarrow G$ is a homotopy equivalence~\cite{Mal45,
	Iwa49} and defines a natural equivalence of groupoids $\cat{Bun}_K(X)\overset\simeq\to \cat{Bun}_G(X)$ on
	spaces $X$, hence a symmetric monoidal equivalence of bordism categories of manifolds with these kinds of
	bundles.
\end{enumerate}
$\Spin_4$ is compact, and the maximal compact subgroup of $\widetilde E_{7(7)}$ is $\SU_8$, so the maximal compact of $H_4$
is the group $\Spin_4\times_{\set{\pm 1}} \SU_8$. Now \cref{stabilization,FHclassification} apply: the
stabilization of $\Spin_4\times_{\set{\pm 1}}\SU_8$ is $\Spin\text{-}\SU_8\coloneqq \Spin\times_{\set{\pm
1}}\SU_8$, and the anomaly field theory is classified by the torsion subgroup of $[\mathit{MT}(\Spin\text{-}\SU_8),
\Sigma^6I_\Z]$, which is determined by $\Omega_5^{\Spin\text{-}\SU_8}$.

In \cref{the_bordism_groups}, we prove $\Omega_5^{\Spin\text{-}\SU_8}\cong\Z/2$, so there is potential for the
anomaly field theory to be nontrivial. The particular anomaly we study is similar to the anomaly classified by $\Omega^{\Spin\text{-}\SU_2}_5=(\Z/2)^2$ where one of the $\Z/2$ is the new $\SU_2$ anomaly in \cite{WWW19}, as well as $\Omega_5^{\Spin\text{-}\Spin_{10}}=\Z/2$. The latter group is applicable to $\Spin_{10}$ GUT-anomalies studied in \cite{WY:21,WY2:21,WY22}, with a more field theoretic point of view.

The quotient map $\Spin\text{-}\SU_8\to\SO\times\SU_8/\set{\pm 1}$ by the central $\set{\pm 1}$ subgroup means that a spin-$\SU_8$ structure on a manifold $M$ induces an orientation on $M$ and a principal $\SU_8/\set{\pm 1}$-bundle $P\to M$, similar to the way a spin\textsuperscript{$c$} structure induces an orientation and a principal $\U_1$-bundle. To lift from an orientation and a $\U_1$-bundle to a spin\textsuperscript{$c$} structure one must trivialize a characteristic class, and likewise for spin-$\SU_8$ structures, as we now prove.
\begin{lem}
\label{smallcoh}
$H^1(B(\SU_8/\set{\pm 1});\Z/2)\cong 0$ and $H^2(B(\SU_8/\set{\pm 1});\Z/2)\cong\Z/2$.
\end{lem}
\begin{proof}
Since $\SU_8/\set{\pm 1}$ is connected, $B(\SU_8/\set{\pm 1})$ is simply connected, and the Hurewicz and universal coefficient theorems then tell us $H^1(B(\SU_8/\set{\pm 1});\Z/2) = 0$. For $H^2$, use that $\pi_2(B(\SU_8/\set{\pm 1}))\cong\pi_1(\SU_8/\set{\pm 1})\cong\Z/2$, because $\SU_8\to\SU_8/\set{\pm 1}$ is a double cover by a simply connected space; then apply the Hurewicz and universal coefficient theorems again.
\end{proof}
\begin{lem}
\label{which_cext}
The central extension
\begin{equation}
\label{spSU8_xt}
    \shortexact*{\set{\pm 1}}{\Spin\times_{\set{\pm 1}}\SU_8}{\SO\times\SU_8/\set{\pm 1}}{}
\end{equation}
is classified by $w_2 + a\in H^2(B(\SO\times\SU_8/\set{\pm 1});\Z/2)$, where $a$ is the nonzero element of
$H^2(BG_8;\Z/2)$.
\end{lem}
\begin{proof}
Since $B(G\times H)\simeq BG\times BH$ for topological groups $G$ and $H$, and since $H^1(B\SO;\Z/2)$ and $H^1(B\SU_8/\set{\pm 1};\Z/2)$ both vanish, the Künneth formula identifies
\begin{equation}
\label{Ksplit}
    H^2(B(\SO\times\SU_8/\set{\pm 1});\Z/2)\cong H^2(B\SO;\Z/2)\oplus H^2(B\SU_8/\set{\pm 1};\Z/2).
\end{equation}
In view of the Künneth splitting~\eqref{Ksplit} and the fact that both $H^2(B\SO;\Z/2)$ and $H^2(B\SU_8/\set{\pm 1};\Z/2)$ are isomorphic to $\Z/2$, to show $x = w_2+a$ it suffices to pull the extension~\eqref{spSU8_xt} back along the inclusions $i_1\colon \SO\to\SO\times\SU_8/\set{\pm 1}$ and $i_2\colon \SU_8/\set{\pm 1}\to\SO\times \SU_8/\set{\pm 1}$ and see that both pullbacks do not split.

After pulling back by $i_1$, we get the extension $1\to\set{\pm 1}\to\Spin\to\SO\to 1$, which is not split, as it is classified by $w_2$; after pulling back by $i_2$, we get the extension $1\to\set{\pm 1}\to\SU_8\to\SU_8/\set{\pm 1}\to 1$; because $\SU_8$ is connected but $\set{\pm 1}$ is not connected, this extension is also not split. Therefore $x = w_2 + a$.
\end{proof}
\begin{lem}
\label{SSU_lift}
Let $M$ be an oriented manifold and $P\to M$ be a principal $\SU_8/\set{\pm 1}$-bundle. The data of a spin-$\SU_8$-structure inducing this orientation and principal $\SU_8/\set{\pm 1}$-bundle is equivalent to a trivialization of $w_2(M) + a(P)$, where $a$ is the unique nonzero element of $H^2(B(\SU_8/\set{\pm 1}); \Z/2)$.\footnote{By a \emph{trivialization} of a characteristic class $a\in H^\ell(X; A)$, we mean data of a null-homotopy of the corresponding map $X\to K(A, \ell)$. Thus a trivialization of $w_1(M)$ on a manifold is equivalent data to an orientation of $M$, a trivialization of $w_2(M)$ on an oriented manifold is equivalent data to a spin structure on $M$, etc.}
\end{lem}
\begin{proof}
Apply the classifying space functor to the short exact sequence~\eqref{spSU8_xt} to obtain a fibration
\begin{equation}\begin{tikzcd}
	{B(\Spin\times_{\set{\pm 1}}\SU_8)} \\
	{B(\SO\times \SU_8/\set{\pm 1})} & {K(\Z/2, 2),}
	\arrow["{w_2 +a}", from=2-1, to=2-2]
	\arrow["\pi", from=1-1, to=2-1]
\end{tikzcd}\end{equation}
implying that for any space $X$ and map $f\colon X\to B(\SO\times\SU_8/\set{\pm 1})$, the data of a lift of $f$ to a map $\widetilde f\colon X\to B(\Spin\times_{\set{\pm 1}}\SU_8)$ such that $\pi\circ\widetilde f = f$ is equivalent data of a null-homotopy of $(w_2 + a)\circ f$. Now let $X = M$ and $f$ be the map classifying $TM$ with its orientation and $P$ to conclude.
\end{proof}
\begin{rem}
If we knew $a$ was $w_2$ of a the associated vector bundle to a representation $\rho$ of $\SU_8/\set{\pm 1}$, then \cref{SSU_lift} would follow from~\cite[Corollary 10.23]{DDHM22}. However, no such $\rho$ exists, as we discuss in \cref{oops_all_spin}.
\end{rem}
%
%
\begin{rem}
Computing bordism groups to determine whether an anomaly is trivial is a well-established technique in the mathematical physics literature: other papers taking this approach include~\cite{Wit86, Kil88, Mon15, Cam17, Mon17, GPW18, Hsi18, STY18, DGL19, Garcia-Etxebarria:2018ajm, MM19, TY19, WW19, WW19c, WWZ19, Witten:2019bou, BLT20, Davighi:2020uab, DL20, DL20c, GOPWW20, HKT20, HTY20, JF20, KPMT20, Tho20, WW20, WW19a, WW20b, FH21, FH21a, DDHM21, DGG21, GP21a, Gri21, Koi21, LOT21, LOT21a, Lee:2020ewl, TY21, Yu:2020twi, WNC21, Davighi:2022icj, Deb22, LY22, Tac22, Yon22}.
\end{rem}

%
%
%
\section{Spectral sequence computation}\label{sec:SSComputation}
The $E_2$ page for U-duality in the Adams spectral sequence is~\cite[Theorem 2.1,\,2.2]{Adams1957/58}
\begin{equation}
\label{the_Adams_sig}
 \Ext^{s,t}_{\cA}(H^*(\mathit{MT}(\Spin\text{-}\SU_8);\Z/2),\Z/2) \Rightarrow \pi_{s-t}(\mathit{MT}(\Spin\text{-}\SU_8))^{\wedge}_2 \cong (\Omega^{\Spin\text{-}\SU_8}_{s-t})^{\wedge}_2\,,
\end{equation}
which converges to the 2-completion of the
desired bordism group via the Pontrjagin-Thom construction.

Let $G_8\coloneqq\SU_8/\set{\pm 1}$. The standard way to tackle Adams spectral sequence questions such as~\eqref{the_Adams_sig} would be to re-express a spin\text{-}$\SU_8$ structure on a vector bundle $E\to M$ as data of a principal $G_8$-bundle $P\to M$ and a spin structure on $E\oplus \rho_P$, where $\rho_P$ is the associated bundle to $P$ and some representation $\rho$ of $G_8$. Once this is done, one invokes a change-of-rings theorem that makes calculating the $E_2$-page of~\eqref{the_Adams_sig} much easier. For several great examples of this technique, see~\cite{Cam17, BC18}.

Unfortunately, this strategy is not available for spin-$\SU_8$ bordism. The reason is that $\rho$, thought of as a map $\rho\colon G_8\to\mathrm O_n$ for some $n$, cannot lift to a map $G_8\to\Spin_n$; if it does, a spin structure on $E\oplus\rho_P$ is equivalent to a spin structure on $E$ by the two-out-of-three property. However, $G_8$ does not have any non-spin representations.
%
%
\begin{thm}[Speyer~\cite{Speyer}]
\label{oops_all_spin}
All representations $\rho\colon G_8 \to \O_n$ lift to $\Spin_n$.
\end{thm}
We give another proof of this fact using $H^*(BG_8;\Z/2)$, which we calculate in \cref{Ugp_coh} in low degrees.
\Cref{Ugp_coh} appears later in this paper, but does not use \cref{oops_all_spin}, so this argument is not
circular.
\begin{proof}
Suppose that a representation $\rho$ of $G_8$ admitting no lift to $\Spin_n$ exists, and let $V\to BG_8$ be its
associated vector bundle; since $\rho$ does not admit a lift to $\Spin_n$, $V$ has no spin structure, so $w_2(V)\ne
0$. The $\cA$-action on the cohomology of the Thom spectrum $(BG_8)^V$ is determined by $H^*(BG_8;\Z/2)$ and
$w_2(V)$ by a formula; we will show below that there is no nonzero choice of $w_2(V)$ such that this formula
satisfies the Adem relations, so no such $\rho$ can exist.

In \cref{smallcoh}, we saw that $H^1(BG_8;\Z/2) = 0$ and $H^2(BG_8;\Z/2) \cong\Z/2$. Let $a$ be the nonzero element
of $H^2(BG_8;\Z/2)$. Since $H^1(BG_8;\Z/2)$ vanishes, $w_1(V) = 0$, and since $H^2(BG_8;\Z/2) = \set{0, a}$ and
$w_2(V)\ne 0$, $w_2(V) = a$. Let $b$ and $d$ be the classes we construct define in $H^*(BG_8;\Z/2)$ in
\cref{Ugp_coh}; then, using the Thom isomorphism and how it affects the $\cA$-module
structure on cohomology (see, e.g., \cite[\S 3.3, \S 3.4]{BC18}), we can compute that if $U$ is the Thom class in
the mod $2$ cohomology of $(BG_8)^V$, then $\Sq^2\Sq^1\Sq^2U = U(ab+d)$, $\Sq^4\Sq^1U = 0$, and there is
no class $x$ with $\Sq^1(Ux) = U(ab+d)$. This contradicts the Adem relation $\Sq^2\Sq^1\Sq^2 = \Sq^4\Sq^1 +
\Sq^1\Sq^4$ that must hold in the cohomology of any space or spectrum.
\end{proof}
%
%
Thus we cannot proceed via the usual change-of-rings simplification, and we must run the Adams spectral sequence over the entire mod $2$ Steenrod algebra $\cA$, which is harder. Similar problems occur in a few other places in the mathematical physics literature, including~\cite{FH21, Deb22}. It would be interesting to find more problems where similar complications occur when trying to work with twisted spin bordism.
%

In order to set up the Adams
computation, a necessary step is to establish the two theorems in
\S\ref{subsec:computingcohomology}  with the goal to give the Steenrod actions on $H^*(B(\Spin\text{-}\SU_8); \Z/2)$. Applying the Thom isomorphism takes care of the rest.
We also detail the simplifications that make working over the entire Steenrod algebra accessible. We refer the reader to \cite{BC18} which highlights many of the computational details of the Adams spectral sequence,
but mainly employs a change of rings to work over $\cA(1)$. 
We start by showing that computing the 2-completion is sufficient for the tangential structure we are considering. 

\subsection{Nothing interesting at odd primes} We will show that the Adams spectral sequence computation that we run which only gives the two torsion part of the anomaly is sufficient for our purposes. 
\begin{prop}
$\Omega_*^{\Spin\text{-}\SU_8}$ has no $p$-torsion when $p$ is an odd prime.
\end{prop}
\begin{proof}
The quotient $\Spin\times\SU_8\to\Spin\text{-}\SU_8$ is a double cover, hence on classifying spaces is
a fiber bundle with fiber $B\Z/2$. $H^*(B\Z/2;\Z/p) = \Z/p$ concentrated in degree $0$, so $B(\Spin\times\SU_8)\to
B(\Spin\text{-}\SU_8)$ is an isomorphism on $\Z/p$ cohomology (e.g.\ look at the Serre spectral
sequence for this fiber bundle). The Thom isomorphism lifts this to an isomorphism of cohomology of the relevant
Thom spectra, and then the stable Whitehead theorem implies that the forgetful map
$\Omega_*^\Spin(B\SU_8)\to\Omega_*^{\Spin\text{-}\SU_8}$ is an isomorphism on $p$-torsion.

The same argument applies to the double cover $\Spin\times\SU_8\to\SO\times\SU_8$, so the $p$-torsion in
$\Omega_*^{\Spin\text{-}\SU_8}$ is isomorphic to the $p$-torsion in $\Omega_*^\SO(B\SU_8)$. Now apply the Atiyah-Hirzebruch
spectral sequence. Averbuh~\cite{Ave59} and Milnor~\cite[Theorem 5]{Mil60} prove there is no $p$-torsion in
$\Omega_\ast^\SO$, and Borel~\cite[Proposition 29.2]{Bor51} shows
there is no
$p$-torsion in $H_*(B\SU_8;\Z)$ and $H_*(B\SU_8;\Z/2)$. Therefore the only way to obtain $p$-torsion in
$\Omega_*^\SO(B\SU_8)$ would be from a differential between free summands, but all free summands in $\Omega_*^\SO$
and $H_*(B\SU_8;\Z)$ are contained in even degrees, so there are no differentials between free summands, and
therefore no $p$-torsion.
\end{proof}

\subsection{Computing the cohomology of $B(\Spin\text{-}\SU_8)$}\label{subsec:computingcohomology}
We first prove \cref{Ugp_coh}, where we compute $H^*(BG_8;\Z/2)$ and its $\cA$-module structure in low degrees. We then use this to compute $H^*(B(\Spin\text{-}\SU_8);\Z/2)$ as an $\cA$-module in low degrees in \cref{thm:cohHstructure}, allowing us to run the Adams spectral sequence in \S\ref{subsec:AdamsComputation}. Our computations make heavy use of the Serre spectral sequence; for more on the Serre spectral sequence and its application to physical problems see \cite{Garcia-Etxebarria:2018ajm,Yu:2020twi, Yu:2021kzr,LY22,Lee:2020ewl,Davighi:2020uab,Davighi:2022icj}. See also Manjunath-Calvera-Barkeshli~\cite[\S D.6]{MCB22}, who perform a related Serre spectral sequence computation to determine some integral cohomology groups of $BG_8$.
\begin{thm}
\label{Ugp_coh}
$H^*(BG_8;\Z/2)\cong\Z/2[a, b, c, d, e, \dots]/(\dots)$ with $\abs a = 2$, $\abs b = 3$, $\abs c =
4$, $\abs d = 5$, and $\abs e = 6$, and there are no other generators or relations below degree $7$. The Steenrod
squares are
\begin{equation}\label{eq:SteenrodSU8/Z2}
\begin{aligned}
	\Sq(a) &= a + b + a^2\\
	\Sq(b) &= b+d + b^2\\
	\Sq(c) &= c+e + \Sq^3(c) + c^2\\
	\Sq(d) &= d + b^2 + \Sq^3(d) + \Sq^4(d) + d^2.
\end{aligned}
\end{equation}
\end{thm}

\begin{proof}
We consider the Serre spectral sequence for the fibration $B\Z/2\to B\SU_8 \to BG_8$ with integer coefficients, which has signature
\begin{equation}\label{eq:reverse}
	E_2^{*,*} = H^*(BG_8; H^*(B\Z/2;\Z)) \Longrightarrow H^*(B\SU_8;\Z).
\end{equation}
Since $BG_8$ is simply connected, we do not need to worry about local coefficients. We know that
$H^*(B\Z/2;\Z)\cong\Z[z]/2z$, where $\abs z = 2$, and Borel~\cite[\S 29]{Bor53a} computed $H^*(B\SU_8;\Z)\cong\Z[ c_2,\dotsc,c_8]$, with
$\abs{c_i} = 2i$, so we may run the spectral sequence in reverse. The $E_2$-page for $\eqref{eq:reverse}$ is:
\begin{gather}\label{eq:SSreverse}
\begin{array}{c|ccccccccc}
    6& z^3& 0&0& \alpha z^3& c_2z^3&\beta z^3&(c_3 z^3, \alpha z^3)\\
    5& 0 &0&0 &0&0&0&0\\
     4 & z^2&0& 0& \alpha z^2& c_2 z^2& \beta z^2&(c_3z^2 , \alpha^2z^2)   \\
     3& 0 &0&0 &0&0&0&0\\
     2& z & 0&0& \alpha z & c_2 z& \beta z& (c_3z ,\alpha^2z) \\
     1& 0 &0&0& 0&0&0&0\\
     0 & 1& 0 & 0 & \alpha & c_2 & \beta & (c_3, \alpha^2)\\
     \hline 
     & 0 & 1 & 2 & 3 & 4 & 5&6 \,.
\end{array}
\end{gather}
As $H^2(B\SU_8;\Z) = 0$, $z\in E_2^{0,2} = H^2(B\Z/2;\Z)$ admits a differential. The only option is a transgressing
$d_3$; let $\alpha\coloneqq d_3(z)$. Since $2z = 0$, $2\alpha = 0$. The Leibniz rule (now with signs) tells us
\begin{equation}
	d_2(z^2) = zd_2(z) + d_2(z)z = 2\alpha z = 0.
\end{equation}
Therefore if $z^2$ admits a differential, the differential must be the transgressing $d_5\colon E_4^{0,4}\to
E_4^{5,0}$, see \eqref{eq:SSreverse}. But $z^2$ does admit a differential. One way to see this is to compute the pullback $H^4(B\SU_8;\Z)\to
H^4(B\Z/2;\Z)$. Since $H^4(B\SU_8;\Z)$ is generated by $c_2$ of the defining representation $\C^8$,
we can restrict
that representation to $\Z/2$ and compute its second Chern class to compute the pullback map. As a representation
of $\Z/2$, $\C^8$ is a direct sum of $8$ copies of the sign representation, so its total Chern class is $c(8\sigma)
= (1 + z)^8$ by the Whitney sum rule, and the $z^2$ term is $\binom 82 z^2$, which is even. Since $2z^2 = 0$, this
implies $c_2$ pulls back to $0$. If $z^2$ did not support a differential, then it would be in the image of this
pullback map, so we have discovered that $z^2$ admits a differential, specifically $d_5$. Let $\beta\coloneqq
d_5(z^2)$. From the spectral sequence we see that $H^4(BG_8; \Z)$ is isomorphic to $H^4(B \SU_8; \bZ)$, and $c_2$ is an element in this cohomology. By using the mod 2 reduction map from $H^4(B \SU_8; \bZ)\to H^4(B \SU_8; \bZ/2)$, and the pullback map induced from $\SU_8\overset{f}{\to} G_8$ we see that $c$ is the mod 2 reduction of  $c_2$ in $ H^4(BG_8;\Z)$. This is summarized in the following diagram:
\begin{equation}
    \begin{tikzcd}
    {c_2\in H^4(BG_8; \Z) } {\arrow[rr, "\mod 2"]}  {\arrow[d, "\cong"]}  &  & {H^4(BG_8; \Z/2)} \arrow[d,"f^*"]  \\
{H^4(B\SU_8; \bZ)} \arrow[rr]                                                                                &  & {H^4(B\SU_8; \bZ/2)}\,.                                      
\end{tikzcd}
\end{equation}
Let $\tilde{c}_2$ be the class in $H^4(B\SU_8;\Z)$ given by the equivalence in the above diagram, and $\overline{{c}}_2$ be the mod 2 reduction with $f^*: c \mapsto \overline{c}_2$. We see that $\Sq^2(\overline{c}_2) = \overline{c}_3$ in $H^6(B\SU_8;\Z/2)$ and therefore get that $e:=\Sq^2(c)$ maps to $\overline{c}_3$. We now consider the Serre spectral sequence for the same fibration but with $\Z/2$ coefficients with signature 
\begin{equation}\label{eq:reverseWithZ2}
   	E_2^{*,*} = H^*(BG_8; H^*(B\Z/2;\Z/2)) \Longrightarrow H^*(B\SU_8;\Z/2).
\end{equation}
The $E_2$-page for \eqref{eq:reverseWithZ2} is:
\begin{gather}\label{eq:SSreverseWithZ2}
\begin{array}{c|ccccccccc}
    5& w^5_1 \\
     4 & w_1^4  \\
     3& w_1^3 \\
     2& w_1^2 \\
     1& w_1 \\
     0 & 1& 0 & a & b & \overline{c}_2 & d & \overline{c}_3\\
     \hline 
     & 0 & 1 & 2 & 3 & 4 & 5&6 \,.
\end{array}
\end{gather}
The classes $\overline{c}_2$ and $\overline{c}_3$ must be detected on the $q=0$ row. As for the rest of the classes in the $p=0$ column we start with $w_1$ which supports a differential so that $d_2(w_1):=a$. Furthermore, $w_1^2$ must also transgress so that $d_3(w_1^2):=b$. We now apply Kudo's transgression theorem \cite{Kud56} which states that Steenrod squares commute with transgression in the Serre spectral sequence to futhermore see that $d_3(w_1^2) = d_3(\Sq^1(w_1))=\Sq^1(a)$. Since we formerly showed that $\overline{c}_2$ already exists, $w_1^4$ must also transgress so that $d_4(w_1^4) = d_3(\Sq^2(w_1^2)) = \Sq^2\Sq^1(a) = \Sq^2(b) =: d $. Then using the Adem relation we get that $\Sq^1(d) = b^2$, and this concludes the proof.
\end{proof}

We now compute $H^*(B(\Spin\text{-}\SU_8);\Z/2)$, which is what  we actually need for U-duality. Recall from
\cref{which_cext} that the central extension
\begin{equation}
\label{Z2F_quotient}
    \shortexact{\Z/2}{\Spin\text{-}\SU_8}{\SO\times G_8},
\end{equation}
which we think of physically as ``quotienting by fermion parity,'' is classified by $w_2 + a\in H^2(B(\SO\times
G_8);\Z/2)$.

Taking classifying spaces, we obtain a fibration
\begin{equation}\label{eq:extensionseq}
    B\bZ/2 \longrightarrow B (\Spin\text{-}\SU_8) \longrightarrow B (\SO \times G_8)\,,
\end{equation}
and we apply the Serre spectral sequence to this fibration
using knowledge of the cohomology of $BG_8$. 

\begin{thm}\label{thm:cohHstructure}
The cohomology ring for $H^*(B(\Spin\text{-}\SU_8); \Z/2)$ considering generators in degree 6 and below is given by  $H^*(B(\Spin\text{-}\SU_8); \Z/2) \cong \bZ/2[a,b,c,w_4,d,e,\ldots]$ with $\abs a = 2$, $\abs b = 3$, $\abs c =
4$, $\abs{w_4} = 4$, $\abs d = 5$, and $\abs e = 6$. 
The map $\Spin\text{-}\SU_8 \rightarrow \SO \times G_8$ induces a quotient map on cohomology with the identification $a=w_2$, $b=w_3$, $ab+d=w_5$ where $w_i$ are the Stiefel-Whitney classes of $B\SO$. 
The Steenrod squares of $a$, $b$, $c$, $d$, and $e$ are given in \eqref{eq:SteenrodSU8/Z2} along with 
\begin{equation}\label{eq:Steenrodw4}
\begin{aligned}
       \Sq^1 w_4 &= ab+d\,,\\
       \Sq^2 w_4 &= a w_4+ x,
\end{aligned}
\end{equation}
where $x$ is a class in $H^6(B(\Spin\text{-}\SU_8);\Z/2)$ that is linearly independent from $aw_4$.
\end{thm}
\begin{proof}
We run the Serre spectral sequence with signature 
\begin{equation}
\label{another_Serre}
    E^{*,*}_2  = H^*(B (\SO \times G_8); H^*(B\bZ/2;\bZ/2)) \Longrightarrow H^*(B (\Spin\text{-}\SU_8); \bZ/2)\,,
\end{equation}
where the $E_2$-page is given by:
\begin{gather}\label{eq:tangential}
\begin{array}{c|ccccccccc}
     5 & t^5&0 \\
     4 & t^4 &0  \\
     3& t^3 &0 \\
     2& t^2 &0&\!\!\!\!(t^2a,t^2a+t^2w_2)& \ldots \\
     1& t &0 & \!\!\!\!\!\!\!(ta,ta\!+\!tw_2)& \!\!\!\!\!\!\!\!(tb,tb\!+\!tw_3)& \ldots \\
      0 & 1& 0 & \!\!\!\!\!(a,a\!+\!w_2) & \!\!\!\!\!\!\!\!(b,b\!+\!w_3) & \!\!\!\!\!\!\begin{pmatrix}a^2,c,w_4,a^2\!+\!w^2_2,\\
      a(a+w_2) \end{pmatrix} & \!\!\!\begin{pmatrix}ab\!+\!d\!+\!w_5,(a\!+\!w_2)(b\!+\!w_3)\\ a(b+w_3),
      b(a+w_2)\end{pmatrix}  \\ 
     \hline 
     & 0 & 1 & 2 & 3 & 4 & 5  \,.
\end{array}
\end{gather}
The $w_i$ are the Stiefel-Whitney classes of $B\SO$, and $t$ is the generator of the cohomology $H^*(B\bZ/2; \bZ/2)$. The differential $d_2:E^{0,1}_2 \to E^{2,0}_2$ hits the class for the extension~\eqref{eq:extensionseq} that gives $\Spin\text{-}\SU_8$,
which is $a+w_2$, and identifies $a=w_2$. Applying the Leibniz rule shows $d_2(t^{2n+1}) = t^{2n}a$, and that $d_2(t^{2n}) = 0$.
We apply Kudo's transgression theorem to see
therefore $d_3:E^{0,2}_3 \to E^{3,0}_3 $ sends $t^2 \mapsto b+w_3$, since the transgressing $d_3$ sends $\Sq^1 t=t^2$ to $\Sq^1 (a+w_1)$.
In total degree 4, there is a likewise $d_5$ differential that takes $t^4$ to $ab+d+w_5$, i.e.\ this differential takes $\Sq^2 t^2$ to $\Sq^2(b+w_3)$.\footnote{Another way to see that $d_5$ identifies $w_5$ with $ab+d$ is from the point of view of $\cA$-modules. As we will see in the proof of \cref{the_bordism_groups}, the only way to have the $\cA$-module in \cref{fig:modulefromU} is if $\Sq^1 (w_4 U) = w_5U$ agrees with $\Sq^2( bU) = (ab+d)U$ }
We see that there is a new class $w_4$ which pulled back from $B\SO$. Applying the Wu formula then establishes the formulas in \eqref{eq:Steenrodw4}, where $x$ is the pullback of $w_6$ by the map $B(\Spin\text{-}\SU_8)\to B\SO$. We still need to show that $x$ and $aw_4$ are linearly independent. In $H^6(B\SO;\Z/2)$, $w_2w_4$ and $w_6$ are linearly independent, so the only way their images in $H^6(B(\Spin\text{-}\SU_8);\Z/2)$ could be linearly dependent is if a differential in the Serre spectral sequence~\eqref{another_Serre} hits $w_6$ or $w_2w_4 + w_6$. As we showed above, the Leibniz rule means all differentials in this spectral sequence are determined by the differentials out of $E_r^{0,q}$, so it suffices to understand these differentials for $0\le q\le 5$, so as to land in total degree $6$. We have already calculated all of these differentials in this proof, and using their values and the Leibniz rule, we see that both $w_6$ and $w_2w_4 + w_6$ survive to the $E_\infty$-page, implying as claimed that $x$ and $aw_4$ are linearly independent.
\end{proof}
\begin{cor}
\label{get_it_all_from_SO}
The composition
\begin{equation}\label{from_BSO}
    H^k(B\SO;\Z/2)\to H^k(B\SO\times BG_8;\Z/2) \to H^k(B(\Spin\text{-}\SU_8);\Z/2)
\end{equation}
is injective for $k\le 6$, and an isomorphism for $0\le k\le 3$ and $k = 5$; for $k = 4$, a complementary subspace to the image of~\eqref{from_BSO} is spanned by $c$, and for degree $6$, a complementary subspace is spanned by $\{e, ac\}$.
\end{cor}
\begin{proof}
We saw in \cref{thm:cohHstructure} that if we quotient by all elements in degrees $7$ and above, the pullback map
\begin{equation}\label{quotient_inj}
    H^*(B\SO\times BG_8;\Z/2) \to H^*(B(\Spin\text{-}\SU_8);\Z/2)
\end{equation}
is surjective with kernel the ideal generated by $\{w_2+a, w_3+b, w_5 + ab + d\}$. Therefore the images of the classes $w_2$, $w_3$, $w_4$, $w_5$, and $w_6$, together with all products of these classes in degrees $6$ and below, are linearly independent in $H^*(B(\Spin\text{-}\SU_8);\Z/2)$. As these classes span the subspace of $H^*(B\SO;\Z/2)$ in degrees $6$ and below, we have proven the lemma.
\end{proof}
\begin{cor}\label{all_SO_2}
The pullback map
\begin{equation}\label{all_SO_but_Thom}
    H^k(\MTSO; \Z/2)\rightarrow H^k(\mathit{MT}(\Spin\text{-}\SU_8); \Z/2),
\end{equation}
is also injective in degrees $6$ and below and an isomorphism in degrees $0\le k\le 3$ and $k = 5$; complementary subspaces to the image of~\eqref{all_SO_but_Thom} are spanned by $Uc$ for $k = 4$ and $\{Uac, Ue\}$ for $k =6$.
\end{cor}
\begin{proof}
This follows from \cref{get_it_all_from_SO} using the naturality of the Thom isomorphism.
\end{proof}

\subsection{The Adams Computation}\label{subsec:AdamsComputation}
In this section, we run the Adams spectral sequence for Spin-$\SU_8$ bordism.
\begin{thm}
\label{the_bordism_groups}
Up to degree 5, the first few groups of ${\Spin\text{-}\SU_8}$ bordism are
\begin{equation}
\begin{aligned}
	\Omega_0^{\Spin\text{-}\SU_8} &\cong \Z\\
	\Omega_1^{\Spin\text{-}\SU_8} &\cong 0\\
	\Omega_2^{\Spin\text{-}\SU_8} &\cong 0\\
	\Omega_3^{\Spin\text{-}\SU_8} &\cong 0\\
	\Omega_4^{\Spin\text{-}\SU_8} &\cong \Z^2\\
	\Omega_5^{\Spin\text{-}\SU_8} &\cong \Z/2.
\end{aligned}
\end{equation}
Treating $d\in H^5(BG_8;\Z/2)$ as a characteristic class, the bordism invariant $(M, P)\mapsto
\int_M d(P)\in\Z/2$ realizes the isomorphism $\Omega_5^{\Spin\text{-}\SU_8}\to\Z/2$.
\end{thm}
Before we begin the proof of \cref{the_bordism_groups}, we need a few lemmas about $\cA$-modules.
\begin{defn}
Let $C\eta$ denote the $\cA$-module $\Sigma^{-2}\widetilde H^*(\mathbb{CP}^2;\Z/2)$,\footnote{This $\cA$-module is
called $C\eta$ because it is the cohomology of the cofiber of the Hopf map $\eta\colon\Sigma\mathbb S\to\mathbb S$,
the nontrivial element of $\pi_1(\mathbb S)$.}
\end{defn}
\begin{lem}
\label{sp_SU_coh}
If one quotients by all
elements in degrees $7$ and above, there is an isomorphism of $\cA$-modules
\begin{equation}\label{eq:MTspin-su8}
	H^*(\mathit{MT}(\Spin \text{-} \SU_8);\Z/2) \cong \textcolor{BrickRed}{\cA\otimes_{\cA(0)} \bZ/2} \oplus
		\textcolor{Green}{\Sigma^4 (\cA\otimes_{\cA(0)} \bZ/2)} \oplus
		\textcolor{MidnightBlue}{\Sigma^4 C\eta} \oplus
		\textcolor{Fuchsia}{\Sigma^5\cA} \oplus P,
\end{equation}
where $P$ contains no nonzero elements in degrees $5$ and below. $U$ generates $\textcolor{BrickRed}{\cA\otimes_{\cA(0)} \bZ/2}$, $Ua^2$ generates $\textcolor{Green}{\Sigma^4 (\cA\otimes_{\cA(0)} \bZ/2)}$, $Uc$ generates $\textcolor{MidnightBlue}{\Sigma^4 C\eta}$.
\end{lem}
\begin{proof}
In this proof, we implicitly quotient graded rings by their ideals of elements in degrees $7$ and above.

Because~\eqref{all_SO_but_Thom} is so close to an isomorphism, it tells us almost all of the $\cA$-module structure we need in terms of the $\cA$-module structure on $H^*(\MTSO; \Z/2)$. Wall~\cite{Wal60} showed that, localized at $p = 2$, $\MTSO$ is a direct sum of Eilenberg-Mac Lane spectra
\begin{equation}
\MTSO_{(2)}\overset\simeq\longrightarrow \textcolor{BrickRed}{H\Z} \vee\textcolor{Green}{\Sigma^4 H\Z} \vee \textcolor{Fuchsia}{\Sigma^5 H\Z/2} \vee \dotsb
\end{equation}
where all unlisted summands are at least $7$-connected, hence do not matter in our calculation. Therefore on cohomology we have
\begin{equation}\label{eq:spectraMTSO}
    H^*(\MTSO;\bZ/2) \cong \textcolor{BrickRed}{H^*(H\bZ)}\oplus \textcolor{Green}{\Sigma^4H^*(H\bZ)} \oplus \textcolor{Purple}{\Sigma^5 H^*(H\bZ/2)} \oplus \ldots\,.
\end{equation}
Serre~\cite{Ser53} showed that $H^*(H \bZ) = \cA\otimes_{\cA(0)} \Z/2$ and $H^*(H \Z/2) = \cA$.

By \cref{all_SO_2}, $\mathcal S\coloneqq \textcolor{BrickRed}{\cA\otimes_{\cA(0)} \bZ/2} \oplus \textcolor{Green}{\Sigma^4 (\cA\otimes_{\cA(0)} \bZ/2)} \oplus \textcolor{Fuchsia}{\Sigma^5\cA}$ is an $\cA$-submodule of $H^*(\mathit{MT}(\Spin \text{-} \SU_8);\Z/2)$; we want to show that $\mathcal S$ is a direct summand, and compute Steenrod squares on a complementary subspace.

To show that $\mathcal S$ is a direct summand, it suffices to compute Steenrod squares on a generating set for a complementary subspace, such as $\{Uc, Ue, Uac\}$, and show that those Steenrod squares do not land in $\mathcal S\setminus 0$. For $Ue$ and $Uac$ this is vacuous: their squares land in degrees $7$ or higher, which we have killed. Similarly, for $Uc$, we only need to compute $\Sq^1$ and $\Sq^2$. We can use the formula for Steenrod squares in a Thom spectrum
\cite[Remark 3.3.5]{BC18}\footnote{The appearance of $-V$, rather than $V$, is because we are studying bordism of spin-$\SU_8$ structures on tangent bundles, rather than on stable normal bundles. See Footnote~\ref{MT_vs_M}.}
\begin{equation}\label{eq:steenrodU}
	\Sq^k(Ux) = \sum_{i=0}^k \Sq^i(U)\Sq^{k-i}(x) = \sum_{i=0}^k Uw_i(-V)\,\Sq^{k-i}(x),
\end{equation}
where $-V\to B(\Spin\text{-}\SU_8)$ is the pullback of the inverse of the tautological vector bundle on $B\SO$. Since $w_1(-V) = 0$ and $\Sq^1(c) = 0$, $\Sq^1(Uc) = 0$; since $w_2(-V) = a$ and $\Sq^2(c) = e$, $\Sq^2(Uc) = U(ac+e)\ne 0$. Therefore  $\mathcal S$ is indeed a direct summand, and the subspace spanned by $\{Uc, U(ac+e)\}$ spans the promised $\textcolor{MidnightBlue}{\Sigma^4 C\eta}$ and $Ue$ spans $P$.
\end{proof}
\begin{lem}
\label{ceta_ext}
$\Ext_\cA(C\eta)$ is isomorphic to $\Z/2[h_0]$ in topological
degree $0$ and vanishes in topological degree $1$.
\end{lem}
\begin{proof}
We use a standard technique: $C\eta$ is part of a short exact sequence of $\cA$-modules
\begin{equation}
\label{ceta_xtn}
    \shortexact{\textcolor{RubineRed}{\Sigma^2 \Z/2}}{C\eta}{\textcolor{Periwinkle}{\Z/2}},
\end{equation}
and a short exact sequence of $\cA$-modules induces a long exact sequence of Ext groups. It is conventional to draw this as if on the $E_1$-page of an Adams-graded spectral sequence; see~\cite[\S 4.6]{BC18} for more information and some additional examples. We draw the short exact sequence~\eqref{ceta_xtn} in \cref{Ceta_LES}, left, and we draw the induced long exact sequence in Ext in \cref{Ceta_LES}, right. Looking at this long exact sequence, there are three boundary maps that could be nonzero in the range displayed; because boundary maps commute with the $\Ext_\cA(\Z/2)$-action, these boundary maps are all determined by \begin{equation}
    \partial\colon\Ext_\cA^{0, 2}(\textcolor{RubineRed}{\Sigma^2 \Z/2})\to \Ext^{1,2}_\cA(\textcolor{Periwinkle}{\Z/2}).
\end{equation}
This boundary map is either $0$ or an isomorphism, and it must be an isomorphism, because
\begin{equation}
    \Ext_\cA^{0, 2}(C\eta) = \Hom_\cA(C\eta, \Sigma^2\Z/2) = 0,
\end{equation}
and if the boundary map vanished, we would obtain $\Z/2$ for this Ext group. Thus we know $\Ext_\cA(C\eta)$ in the range we need.
\end{proof}

\begin{figure}[h!]
\begin{subfigure}[c]{0.4\textwidth}
\begin{tikzpicture}[scale=0.6]
        \sqtwoR(0, 0);
        \begin{scope}[RubineRed]
                \foreach \x in {-4, 0} {
                        \tikzpt{\x}{2}{}{};
                }
                \draw[->, thick] (-3.75, 2) -- (-0.5, 2);
        \end{scope}
        \begin{scope}[Periwinkle]
                \foreach \x in {0, 4} {
                        \tikzpt{\x}{0}{}{};
                }
                \draw[->, thick] (0.5, 0) -- (3.75, 0);
        \end{scope}
        \node[below=2pt] at (-4, 0) {$\Sigma^2\Z/2$};
        \node[below=2pt] at (0, 0) {$C\eta$};
        \node[below=2pt] at (4, 0) {$\Z/2$};
\end{tikzpicture}
\end{subfigure}
\hfill
\begin{subfigure}[c]{0.55\textwidth}
\includegraphics[width=.7\textwidth]{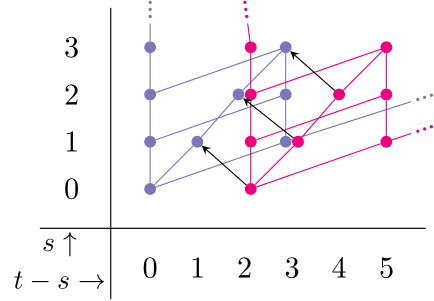}
\end{subfigure}
\caption{Left: the short exact sequence~\eqref{ceta_xtn}. Right: the induced long exact sequence in Ext groups.
These diagrams are part of the proof of \cref{ceta_ext}.}
\label{Ceta_LES}
\end{figure}

\begin{proof}[Proof of \cref{the_bordism_groups}]
To compute the $E_2$-page of the Adams spectral sequence we need to know $\Ext$ of each summand in
\eqref{eq:MTspin-su8} ($\Ext(\text{--})$ means $\Ext_{\cA}^{*,*}(\text{--};\Z/2)$.) Quotienting by elements in degrees $7$ and above does not change the Ext groups in topological degrees less than $6$; one way to see this is to run the long exact sequence in Ext groups induced by the short exact sequence
\begin{equation}
    \shortexact{\tau_{\ge 7}M}{M}{M/\tau_{\ge 7}M},
\end{equation}
where $M$ is an $\cA$-module and $\tau_{\ge 7}M$ is the submodule of $M$ of classes in degree $7$ and above. Therefore, since we are only interested in topological degrees $5$ and below, we need not distinguish $M$ from $M/\tau_{\ge 7}M$ for the purpose of computing Ext groups.

By using the change of rings
theorem \cite[Section 4.5]{BC18}, we get $\Ext_{\cA}(\cA\otimes_{\cA(0)}\Z/2, \Z/2) = \Ext_{\cA(0)}(\Z/2,\Z/2)$, and since $\cA(0)$ only includes $\Sq^1$, this just gives $\Z/2[h_0]$ \cite[Remark 4.5.4]{BC18}, where $h_0\in\Ext^{1,1}$. The same logic applies to $\Ext_\cA(\textcolor{Green}{\Sigma^4(\cA\otimes_{\cA(0)}\Z/2)})$, and $\Ext_\cA(\textcolor{Fuchsia}{\Sigma^5\cA})$ contributes a $\bZ/2$ in degree 5. 

\begin{figure}
    \begin{subfigure}[c]{0.35\textwidth}
\begin{tikzpicture}[scale=0.6, every node/.style = {font = \tiny}]
	\foreach \y in {0, ..., 7} {
		\node at (-2, \y) {$\y$};
	}
	\begin{scope}[BrickRed]
		\Mzero{2}{0}{$U$};
		\tikzpt{2}{4}{}{};
		\tikzpt{2}{6}{}{};
		\tikzpt{2}{7}{}{};
		\draw (2,3.7) node {$Uw_4$};
		\draw (2,5.3) node {$\alpha$};
		\draw (2,4)--(2.5,4)--(2.5,0)--(2,0);
  \draw (2,4+2)--(2.5+.3,4+2)--(2.5+.3,0+2)--(2,0+2);
        \draw (2,7)--(1.5-.4,7)--(1.5-.4,3)--(2,3);

		\sqone(2, 4);
	       \sqtwoR(2, 4);
	       	\sqone(2, 6);
	\end{scope}
\end{tikzpicture}
\end{subfigure}
    \caption{The only relevant higher Steenrod operation in this degree is $\Sq^4$, which acts on $U$ to give $Uw_4$. This is connected to $\alpha=(ab+d)U$ by $\Sq^1$. (This module continues past degree $7$, but we do not need any higher-degree information.)}
    \label{fig:modulefromU}
\end{figure}

Compiling the information of $\Ext$ on \eqref{eq:MTspin-su8} we draw the $E_2$-page
of the Adams spectral sequence through topological degree $5$ in \cref{the_adams_bit}.
\begin{figure}[h!]
\includegraphics[width=.35\textwidth]{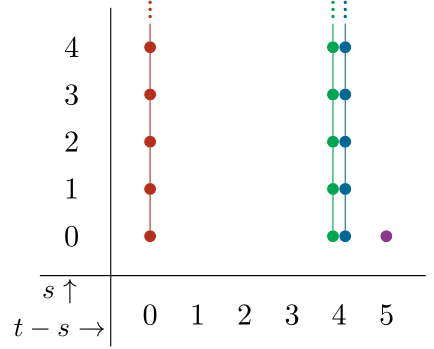}
\caption{The $E_2$-page of the Adams spectral sequence computing
$\Omega_*^{\Spin \text{-} \SU_8}$\,. }
\label{the_adams_bit}
\end{figure}

In this range, the only differentials that could be nonzero go from the $5$-line to the $4$-line. Usually we would
need to know the $6$-line in order to determine if there are any differentials from the $6$-line to the $5$-line,
so that we could evaluate $\Omega_5^{\Spin \text{-} \SU_8}$, but the $5$-line is concentrated in filtration zero, and all
Adams differentials land in filtration $2$ or higher, so what we have computed is good enough.

Returning to the differentials from the $5$-line to the $4$-line: Adams
differentials must commute with the action of $h_0$ on the $E_r$-page, and $h_0$ acts by $0$ on the $5$-line but
injectively on the $4$-line, so these differentials must also vanish. Thus the spectral sequence collapses giving
the bordism groups in the theorem statement. The fact that $\Omega_5^{\Spin\text{-}\SU_8}\cong\Z/2$ is detected by $\int d$
follows from the fact that its image in the $E_\infty$-page is in Adams filtration zero, corresponding to Ext of
the free $\textcolor{Fuchsia}{\Sigma^5\cA}$ summand generated by $Ud$; see~\cite[\S 8.4]{FH21}.
\end{proof}


\subsection{Determining the Manifold Generator}
\label{mf_gen}

We now determine the generator of~$\Omega_5^{\Spin\text{-}\SU_8}\cong\Z/2$. We start by considering a map $\widetilde\Phi\colon\SU_2\to\SU_8$ sending a matrix $A$ to its fourfold block sum $A\oplus A\oplus A\oplus A$. This sends $-1\mapsto -1$, so $\widetilde\Phi$ descends to a map
\begin{equation}
	\Phi\colon \SO_3 = \SU_2/\set{\pm 1}\longrightarrow \SU_8/\set{\pm 1} = G_8.
\end{equation}
Recall that $H^*(B\SO_3;\Z/2)\cong\Z/2[w_2, w_3]$ and that there are three classes $a$, $b$, and $d$ in \phantom{}
$H^*(B(\SU_8/\set{\pm
1});\Z/2)$.
\begin{lem}
\label{pullback_SO3}
$\Phi^*(a) = w_2$, $\Phi^*(b) = w_3$, and $\Phi^*(d) = w_2w_3$.
\end{lem}
This will imply that to find a generator, all we have to do is find a closed, oriented $5$-manifold $M$ with a
principal $\SO_3$-bundle $P\to M$ with $w_2(M) = w_2(P)$ and $w_2(P)w_3(P) \ne 0$. This is easier than directly
working with $G_8$!
\begin{proof}[Proof of \cref{pullback_SO3}]
Once we show $\Phi^*(a) = w_2$, we're done:
\begin{subequations}
\begin{equation}
	\Phi^*(b) = \Phi^*(\Sq^1(a)) = \Sq^1(\Phi^*(a)) = \Sq^1(w_2) = w_3,
\end{equation}
where the last equal sign follows by the Wu formula. In a similar way
\begin{equation}
	\Phi^*(d) = \Phi^*(\Sq^2(b)) = \Sq^2(\Phi^*(b)) = \Sq^2(w_3) = w_2w_3,
\end{equation}
\end{subequations}
again using the Wu formula. So all we have to do is pull back $a$.

Consider the commutative diagram of short exact sequences
\begin{equation}\begin{tikzcd}
	1 & {\Z/2} & {\SU_2} & {\SO_3} & 1 \\
	1 & {\Z/2} & {\SU_8} & {G_8} & {1.}
	\arrow[from=1-1, to=1-2]
	\arrow[from=2-1, to=2-2]
	\arrow[from=1-3, to=1-4]
	\arrow[from=1-4, to=1-5]
	\arrow[from=2-2, to=2-3]
	\arrow[from=2-3, to=2-4]
	\arrow[from=2-4, to=2-5]
	\arrow[from=1-2, to=1-3]
	\arrow[Rightarrow, no head, from=1-2, to=2-2]
	\arrow["\widetilde\Phi", from=1-3, to=2-3]
	\arrow["\Phi", from=1-4, to=2-4]
\end{tikzcd}
\end{equation}
Taking classifying spaces, this shows that the pullback of the fiber bundle $B\Z/2\to B\SU_8\to
BG_8$ along the map $\Phi\colon B\SO_3\to BG_8$ is the fiber bundle
$B\Z/2\to B\SU_2\to B\SO_3$. We therefore obtain a map between the Serre spectral sequences computing the mod $2$
cohomology rings of $B\SU_2$ and $B\SU_8$, and it is an isomorphism on $E_2^{0,*}$, i.e.\ on the cohomology
of the fiber.

Both $B\SU_2$ and $B\SU_8$ are simply connected, so $H^1(\text{--};\Z/2)$ vanishes for both spaces. Therefore in
both of these Serre spectral sequences, the generator of $E_2^{0,1} = H^1(B\Z/2;\Z/2)$ must admit a
differential. The only differential that can be nonzero is the transgressing $d_2\colon E_2^{0,1}\to E_2^{2,0}$; in
$E_2(\SU_8)$, we saw in \eqref{eq:SSreverseWithZ2} that $d_2(w_1) = a$, and in $E_2(\SU_2)$, $d_2(x) = w_2$, because $w_2$ is the
only nonzero element of $E_2^{2,0} = H^2(B\SO_3;\Z/2)$. Since the pullback map of spectral sequences commutes with
differentials, this means $\Phi^*(a) = w_2$ as desired.
\end{proof}

Now let $W\coloneqq\SU_3/\SO_3$, which is a closed, oriented $5$-manifold called the \term{Wu manifold}, and let
$P\to W$ be the quotient $\SU_3\to\SU_3/\SO_3$, which is a principal $\SO_3$-bundle.
\begin{rem}
The name ``Wu manifold'' for $W$ is due to Barden~\cite[\S 1]{Bar65}. The name refers to work of Wu~\cite{Wu50}
(see also Dold~\cite{Dol56}), though the manifold Wu studied is bordant to, but not diffeomorphic to, $W$.
See~\cite[Footnote 9]{LOT21} for more on the history of $W$.
\end{rem}
We will need to know $H^*(W;\Z/2)$ and the Stiefel-Whitney classes $w(TM)$ and $w(P)$. Most of this is in the
literature, but scattered, so we give a complete account in \cref{H2Wu,wu_coh,wu_facts}.
\begin{lem}[{Barden~\cite[Lemma 1.1]{Bar65}}]
\label{H2Wu}
$W$ is connected and simply connected, and $H^2(W;\Z/2)\cong\Z/2$.
\end{lem}
\begin{proof}
Connectivity follows from connectivity of $\SU_3$.

In view of the Hurewicz and universal coefficient theorems, it suffices to show $\pi_1(W) \cong 0$ and
$\pi_2(W)\cong\Z/2$. To show this, use the long exact sequence in homotopy groups associated to the fiber bundle
$\SO_3\to\SU_3\to W$: for $\pi_1(W) = 0$, use exactness at
\begin{equation}
	\dotsb\longrightarrow \pi_1(\SU_3)\longrightarrow \pi_1(W) \longrightarrow \pi_0(\SO_3)\longrightarrow\dotsb,
\end{equation}
together with the fact that $\SO_3$ is connected and $\SU_3$ is simply connected. To show $\pi_2(W) \cong\Z/2$, use
exactness at
\begin{equation}
	\dotsb\longrightarrow \pi_2(\SU_3)\longrightarrow \pi_2(W) \longrightarrow \pi_1(\SO_3)\longrightarrow
	\pi_1(\SU_3)\longrightarrow \dotsb
\end{equation}
and the isomorphisms $\pi_1(\SO_3)\cong\Z/2$, $\pi_1(\SU_3)\cong 0$, and $\pi_2(\SU_3)\cong 0$; the fact that
$\pi_2(G)$ vanishes for $G$ a connected, simply connected, simple Lie group is due to Bott~\cite[Theorem A]{Bot56}.
\end{proof}
\begin{cor}
\label{wu_coh}
$H^*(W;\Z/2)\cong\Z/2[z_2, z_3]/(z_2^2, z_3^2)$ with $\abs{z_2} = 2$ and $\abs{z_3} = 3$. 
\end{cor}
\begin{proof}
This is the unique cohomology ring satisfying both \cref{H2Wu} and Poincaré duality.
\end{proof}
\begin{prop}
\label{wu_facts}
In $H^*(W;\Z/2)$, the Steenrod squares are
\begin{equation}
\begin{aligned}
	\Sq(z_2) &= z_2 + z_3\\
	\Sq(z_3) &= z_3 + z_2z_3,
\end{aligned}
\end{equation}
and the Stiefel-Whitney class is $w(W) = 1 + z_2 + z_3$. Moreover, $w(P) = 1 + z_2 + z_3$. Thus $w_2(M) = w_2(P)$, so $W$ with $G_8$-bundle induced from $P$ has a spin\text{-}$\SU_8$ structure, and $w_2(P)w_3(P)\ne
0$, meaning $(W, P)$ is our sought-after generator of $\Omega_5^{\Spin\text{-}\SU_8}$.
\end{prop} 
The class $w_2(W)$ was computed by Barden~\cite[Lemma 1.1(v)]{Bar65}. The rest of $w(TM)$ was found by
Landweber-Stong (see~\cite[Appendix D]{LM89}), and the Steenrod squares in $H^*(W;\Z/2)$ by Calabi (see~\cite[\S
3]{Flo73}). Finally, Chen~\cite[\S 6]{Che17} computes $w(P)$.
\begin{proof}
Once we know the Steenrod squares are as claimed, the total Stiefel-Whitney class of $W$
follows from Wu's theorem as follows. The second Wu class $v_2$ is defined to be the Poincaré dual of the map
\begin{equation}
	x\mapsto \int_W \Sq^2(x)\colon H^3(W;\Z/2)\to H^5(W;\Z/2)\to\Z/2
\end{equation}
via the Poincaré duality identification $H^2(W;\Z/2)\cong (H^3(W;\Z/2))^\vee$. Wu's theorem shows that $v_2 = w_2 +
w_1^2$, so since $H^1(W;\Z/2) = 0$ (\cref{H2Wu}), $w_1 = 0$ and $w_2 = v_2$. Since $\Sq^2(z_3) = z_2z_3$, $w_2\ne
0$, so it must be $z_2$. Then $w_3 = \Sq^1(w_2) = z_3$; $w_4$ is trivial for degree reasons (\cref{wu_coh}); and
$w_5 = 0$ follows from the Wu formula calculating $\Sq^1(w_4)$.

Consider the Serre spectral sequence for the fiber bundle
\begin{equation}
\label{Wu_fiber}
\begin{tikzcd}
	{\SO_3} & {\SU_3} \\
	& W,
	\arrow[from=1-1, to=1-2]
	\arrow[from=1-2, to=2-2]
\end{tikzcd}
\end{equation}
which has signature
\begin{equation}
\label{mfld_serre}
	E_2^{*,*} = H^*(W; H^*(\SO_3;\Z/2)) \Longrightarrow H^*(\SU_3;\Z/2).
\end{equation}
A priori we must account for the action of $\pi_1(W)$ on $H^*(\SO_3;\Z/2)$, but by \cref{H2Wu}, $W$ is simply
connected; therefore we do not have to worry about this. As manifolds, $\SO_3\cong\RP^3$, so $H^*(\SO_3;\Z/2)\cong\Z/2[x]/(x^4)$. Also, $H^*(\SU_3;\Z/2)\cong \Z/2[b_2,
b_3]/(b_2^2, b_3^2)$, with $\abs{b_2} = 3$ and $\abs{b_3} = 5$~\cite[\S 8]{Bor54}, allowing us to draw the
$E_2$-page of the Serre spectral sequence~\eqref{mfld_serre}, which appears in~\cref{pullback_SWs}, left.

The fibration~\eqref{Wu_fiber} pulls back from the universal
$\SO_3$-bundle $\SO_3\to E\SO_3\to B\SO_3$ via the classifying map $f_P$ for $P$, inducing a map of Serre spectral
sequences that commutes with the differentials. We draw this map in \cref{pullback_SWs}. This map is an isomorphism on the line $E_2^{0, *}$, so $x\in
E_2^{0,1}(\SU_3)$ pulls back from the generator $x\in E_2^{0,1}(E\SO_3)$ --- and therefore $d_2(x) = z_2$ pulls
back from a class in $E_2^{2,0} = H^2(B\SO_3;\Z/2)$. The only nonzero class in that degree is $w_2$, so $f_P^*(w_2)
= z_2$, i.e.\ $w_2(P) = z_2$.

The Leibniz rule that in the Serre spectral sequence for $\SU_3$, $d_2(x^2) = 2xd_2(x) = 0$. But because
$H^2(\SU_2;\Z/2) = 0$, $x^2$ must support some nontrivial differential. Apart from $d_2$, the only option is the transgressing
$d_3\colon E_3^{0,2}\to E_3^{3,0}$, forcing $d_3(x^2) = z_3$. A similar argument in the Serre spectral sequence for
$E\SO_3$ shows that in that spectral sequence, $d_3(x^2) = w_3$; therefore $f_P^*(w_3) = z_3$ and $w_3(P) = z_3$.
Pullback commutes with Steenrod squares and $\Sq^1(w_2) = w_3$, so $\Sq^1(z_2) = z_3$. Finally, $f_P^*(w_2w_3) =
z_2z_3$, and the Wu formula implies $\Sq^2(w_3) = w_2w_3$, so $\Sq^2(z_3) = z_2z_3$. We have computed all the
Steenrod squares that could be nonzero for degree reasons, and along the way shown $w(P) = 1 + z_2 + z_3$: the
higher-degree Stiefel-Whitney classes of a principal $\SO_3$-bundle vanish.
\end{proof}

\begin{figure}[h!]
\includegraphics[width=.85\textwidth]{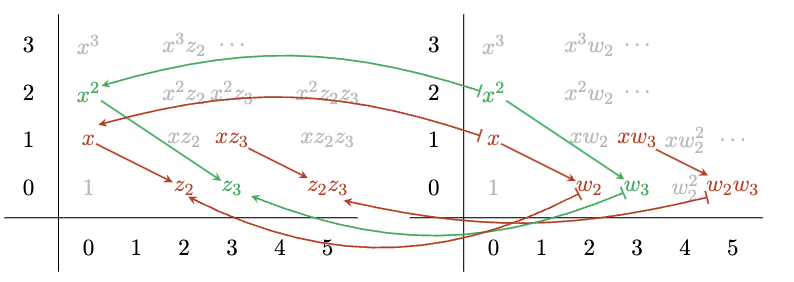}
\caption{
The fiber bundle $\SO_3\to\SU_3\to W$ pulls back from the universal $\SO_3$-bundle $\SO_3\to E\SO_3\to B\SO_3$, inducing a map of Serre spectral sequences. This map commutes with differentials and is the identity on $E_2^{0, \bullet} = H^*(\SO_3;\Z/2)$, allowing us to conclude that $\textcolor{BrickRed}{w_2}$ pulls back to $\textcolor{BrickRed}{z_2}$, $\textcolor{Green}{w_3}$ pulls back to $\textcolor{Green}{z_3}$, and $\textcolor{BrickRed}{w_2w_3}$ pulls back to $\textcolor{BrickRed}{z_2z_3}$. This is a picture proof of part of \cref{wu_facts}.}
\label{pullback_SWs}
\end{figure}

\section{Evaluating on the Anomaly}\label{sec:EvalAnomaly}
With the knowledge of the generating manifold for the $\bZ/2$ in degree 5 as the Wu manifold, we can consider evaluating the anomaly of the theory with the field content given in \S\ref{sec:BordismSetup}. Since $G_8$ acts trivially on the scalars and the graviton only the remaining three fields could have anomalies. 


\begin{defn}
The global anomaly for a fermion on a Riemannian manifold $M$  in a 
 a representation $\textit{R}$ coupled to background $G$ gauge field is given by an invertible field theory with partition function the exponential of an $\eta$-invariant of the Dirac operator, $\eta_{M,\textit{R}}(\mathcal{D}_{\text{Dirac}})$~\cite[Section 4.3]{Witten:2019bou}. 
\begin{itemize}
    \item For gauginos it is given by 
 $\alpha_{1/2} = \exp(\pi i \eta_{M,\textit{R}}(\mathcal{D}_{\text{Dirac}})/2)$ \cite{Witten:2015aba,FH21}.
 \item For gravitinos it is given by   $\alpha_{3/2} = \exp(\pi i \eta_{\text{gravitino}}/2)$ where  \begin{equation}\label{eq:gravitinio}
    \eta_{\text{gravitino}} = \eta_{M,\textit{R}}(\mathcal{D}_{\text{Dirac}\times TM})-2\eta_{M,\textit{R}}(\mathcal{D}_{\text{Dirac}})\,,
\end{equation}
and  $\eta_{M,\textit{R}}(\mathcal{D}_{\text{Dirac}\times TW})$ is the Dirac operator acting on the spinor bundle tensored with the tangent bundle.\footnote{Different sources in the literature give slightly different formulas for~\eqref{eq:gravitinio}, all of the form $\eta_{\text{gravitino}} = \eta_{M,\textit{R}}(\mathcal{D}_{\text{Dirac}\times TM})-k\eta_{M,\textit{R}}(\mathcal{D}_{\text{Dirac}})$, where $k$ is an even integer depending on the dimension of the theory. For a four-dimensional theory, $k = 2$: this is the unique value of $k$ consistent with perturbative anomaly cancellation for the U-duality symmetry as in~\cite{Mar85} and agrees with the value given in~\cite[\S 3.2]{FH21} following~\cite[Appendix A]{Freed:2002qp}.}
\end{itemize}
\end{defn}
For the remainder of the paper we will sometimes drop the $M$ subscript label when there is no confusion that we are considering the generating manifold. We also replace $TM$ by $TW$ where $W$ is the Wu manifold.
\begin{lem}\label{lem:sumofRep}
 If $\mathit{R} = \sum_{i} \mathit{R}_i$ then $\eta_{\sum_{i} \textit{R}_i}(\mathcal{D}_{\text{Dirac}})= \sum_i \eta_{\textit{R}_i}(\mathcal{D}_{\text{Dirac}})$.
\end{lem}

The anomaly for the vector boson is not given in terms of an $\eta$-invariant, but we assume that it is also an invertible theory, and we show that it also vanishes.
The next section is dedicated to showing:

\begin{thm}\label{prop:anomalyVanish}
The total anomaly (global and perturbative) of 4d $\mathcal{N}=8$ supergravity arising from the gaugino, vector boson, and gravitino, vanishes on the Wu manifold. 
\end{thm}

\subsection{Evaluating on the Wu manifold}
The full anomaly denoted by $\alpha$ can be written  schematically as 
\begin{equation}\label{eq:totalanomaly}
    ``\alpha = \alpha^{\text{pert}}_{1/2} \otimes  \alpha^{\text{pert}}_{1} \otimes  \alpha^{\text{pert}}_{3/2} \otimes  \alpha^{\text{np}}_{1/2} \otimes  \alpha^{\text{np}}_{1}\otimes  \alpha^{\text{np}}_{3/2}\,''
\end{equation}
where we have split up each part of the perturbative and nonperturbative anomaly coming from the gaugino, vector boson, and gravitino. Technically speaking, separating the anomaly in this way  is not something that can be done canonically. By \eqref{eq:SESanomaly} the nontopological part arises as a quotient of the invertible theory by the topological theories.
We write the anomaly in such a way in order to make it organizationally more clear. 

The vector bosons can be defined without choosing a spin structure, and therefore the partition function of their anomaly field theory factors through the quotient by fermion parity. That is, the tangential structure is
    \begin{equation}
        \SO\times G_8 = (\Spin\text{-}\SU_8)/\set{\pm 1}\,.
    \end{equation}
We will proceed in understanding the perturbative anomalies by isolating $\alpha^{\text{pert}}_{1}$.
\begin{lem}
The perturbative anomaly for the vector bosons independently vanishes.
\end{lem}
\begin{proof}
With the knowledge that the manifold generator for the anomaly is the Wu manifold, we will further restrict to the $\SO_3$ inside of $G_8$; we are left to computing $\Omega^{\SO}_6(B\SO_3) \otimes \mathbb{Q}$\,, which isolates the free summand. For the degree we are after, we can compute the bordism group via the AHSS. We take the $E^2$ page of
\begin{equation}
    E_{p,q}^2 = H_p(B\SO_3,\Omega^{\SO}_q(\pt)) \Longrightarrow \Omega^{\SO}_6(B\SO_3)\,,
\end{equation}
where 
\begin{equation}\label{eq:SObordism}
    \Omega_{*}^{\SO}(\pt) = \{\,\bZ\,,\,0\,,\,0\,,\,0\,,\bZ\,,\, \bZ/2\,,\,0\,, \ldots\, \}\,,
\end{equation}
and tensor with $\mathbb{Q}$. This will kill the torsion parts that arise on the line $q=5$, and the result is equivalent to the $E_\infty$ page, as all differentials vanish. The $E_\infty$ page of $\Omega^{\SO}_6(B\SO_3) \otimes \mathbb{Q}$ is given by 
\begin{gather}
\begin{array}{c|ccccccccc}
    6& 0\\
    5& 0 &0\\
     4 & \mathbb{Q}&0& 0& 0&\mathbb{Q} &0 \\
     3& 0 &0&0 &0&0&0\\
     2& 0 & 0&0& 0 & 0& 0 \\
     1& 0 &0&0& 0&0&0\\
     0 & \mathbb{Q}& 0 & 0 & 0 & \mathbb{Q} & 0 & 0\\
     \hline 
     & 0 & 1 & 2 & 3 & 4 & 5&6 \,
\end{array}
\end{gather}
We see that the perturbative anomaly of the vector boson vanishes.  
\end{proof}

\begin{cor}
\label{spin_pertur}
The perturbative anomalies from the fractional spin particles vanish on their own.
\end{cor}
Having established this corollary, we may now pullback the anomaly in \eqref{eq:totalanomaly} to the nonperturbative part, and the equation becomes literally true.

The $\eta$-invariant for the contributions in $\alpha^{\text{np}}_{1/2}\otimes  \alpha^{\text{np}}_{3/2}$ is therefore a bordism invariant.
\begin{lem}
The spin-$\SU_8$ bordism invariant $\alpha_{1/2}^{\mathrm{np}}\otimes\alpha_{3/2}^{\mathrm{np}}\colon\Omega_5^{\Spin\text{-}\SU_8}\to\C^\times$ vanishes.
\end{lem}
\begin{proof}
It suffices to check on a generating set for $\Omega_5^{\Spin\text{-}\SU_8}$, and in \S\ref{mf_gen} we showed that the Wu manifold $W$ generates. In particular, the spin-$\SU_8$ structure on $W$ is induced from a spin-$\SU_2$ structure, so the $\eta$-invariants that make up the partition functions of $\alpha_{1/2}^{\mathrm{np}}$ and $\alpha_{3/2}^{\mathrm{np}}$ can be computed by pulling back representations from $\SU_8$ to $\SU_2$, then evaluating the corresponding $\SU_2$ $\eta$-invariants on $W$.

We will show that there are representations $V_1$ and $V_2$ of $\SU_2$ such that $\alpha_{1/2}^{\mathrm{np}}(W) = \eta_{W, 2V_1}(\mathcal D_{\mathrm{Dirac})}$ and $\alpha_{3/2}^{\mathrm{np}}(W) = \eta_{W, 2V_2}(\mathcal D_{\mathrm{Dirac})}$. By \cref{lem:sumofRep}, this means
\begin{equation}
\label{its_a_square}
    (\alpha_{1/2}^{\mathrm{np}}\otimes\alpha_{3/2}^{\mathrm{np}})(W) = \exp\left(\frac{i\pi}{2} \eta_{W, V_1\oplus V_2}(\mathcal D_{\mathrm{Dirac}})\right)^2.
\end{equation}
The quantity $\exp((i\pi/2) \eta_{W, V_1\oplus V_2}(\mathcal D_{\mathrm{Dirac}}))$ is a bordism invariant of spin-$\SU_2$ manifolds: its perturbative part (i.e.\ its image in $\Hom(\Omega_6^{\Spin\text{-}\SU_2}, \Z)$) vanishes, because that invariant is $1/2$ of the perturbative part of $\alpha_{1/2}\otimes\alpha_{3/2}$, which we showed vanishes in \cref{spin_pertur} (and is $\Z$-valued, so $1/2$ of it is uniquely defined if it exists). The equation $2[W] = 0$ in $\Omega_5^{\Spin\text{-}\SU_8}$ pulls back to imply $2[W] = 0$ in $\Omega_5^{\Spin\text{-}\SU_2}$,\footnote{Alternatively, the equation $2[W] = 0$ in spin-$\SU_2$ bordism follows from calculations of Freed-Hopkins~\cite[Figure 5, case $s = 4$]{FH21a}.} so since $\alpha_{1/2}^{\mathrm{np}}\otimes\alpha_{3/2}^{\mathrm{np}}(W)$ is the square of another spin-$\SU_2$ bordism invariant on $W$ by~\eqref{its_a_square}, we conclude $\alpha_{1/2}^{\mathrm{np}}\otimes\alpha_{3/2}^{\mathrm{np}}(W) = 0$.

The rest of this proof is devoted to finding $V_1$ and $V_2$.
%
We consider how  $\textbf{56}$, $\textbf{28}$, and $\textbf{8}$ split via our fourfold embedding of $\SU_2$ into $G_8$ for the Wu manifold. We see that $\textbf{56}$ gives the dimension of the alternating three forms in 8-dimensions, $\textbf{28}$ the dimension of alternating two forms, and $\textbf{8}$ is the defining representation. 
 The branchings are given by
 \begin{align}\label{eq:fermionbranching}
      \textbf{56} &\rightarrow  2(10 \times \textbf{2}+2\times \textbf{4}), \\ \label{eq:vectorbranching}
      \textbf{28} &\rightarrow  2(3 \times \textbf{3}+5\times \textbf{1}),\\
      \textbf{8} &\rightarrow 4 \times \textbf{2}\,, \label{eq:branching}
 \end{align}
  where the right hand side is in terms of $\mathfrak{su_2}$ representations. In increasing numerical order, they are the trivial, defining, adjoint, and $\mathrm{Sym}^4$ representation.
  To show this, notice that $\boldsymbol 8$ splits as $V^{\oplus 4}$ when we pull back to $\SU_2$, where $V = \boldsymbol 2$.
  We then consider the ways of splitting the alternating three forms. This can be done as 
  \begin{equation}
      \Lambda^2 V \otimes \Lambda^1 V \otimes \Lambda^0 V \otimes \Lambda^0 V = \mathbb{C} \otimes V \otimes \mathbb{C} \otimes \mathbb{C}\,
  \end{equation}
  in 12 ways, essentially partitioning $3$ into a sum of length 2. The $\mathbb{C}$ for both $\Lambda^2 V$ and $\Lambda^0 V$ show that they are isomorphic as representations. It can also split into 
  \begin{equation}
      \Lambda^1 V \otimes \Lambda^1 V \otimes \Lambda^1 V \otimes \Lambda^0 V = V \otimes V \otimes V \otimes \mathbb{C}
  \end{equation}
  in 4 ways. The fact that the third tensor product of the defining representation decomposes as $\textbf{2} \otimes \textbf{2} \otimes \textbf{2} = \textbf{2}+\textbf{2}+\textbf{4}$, gives us \eqref{eq:fermionbranching}.
  Similarly, the two forms can be split into 
  \begin{equation}
      \Lambda^2 V \otimes \Lambda^0 V \otimes \Lambda^0 V \otimes \Lambda^0 V \quad \text{and} \quad  \Lambda^1 V \otimes \Lambda^1 V \otimes \Lambda^0 V \otimes \Lambda^0 V
  \end{equation}
  in 4 ways and 6 ways, respectively. The fact that $\textbf{2} \otimes \textbf{2}  = \textbf{1} +\textbf{3}$, establishes \eqref{eq:vectorbranching}.
  
  
  As a spin 3/2 particle, the gravitino contains a spinor index as well as a Lorentz index, 
therefore in order to use \eqref{eq:gravitinio} for the anomaly, we need to use the fact that the tangent bundle of the Wu manifold is an associated bundle.
\begin{lem}
The tangent bundle of the Wu manifold $W$ is given by 
\begin{equation*}
    TW=\SU_3 \times _{\SO_3}\frac{\su_3}{\mathfrak{so}_3}\,.
\end{equation*}
\end{lem}
\begin{proof}
The fact that the Wu manifold is a homogeneous space allows us to use the following general procedure to construct its tangent bundle.
For $H\subset G$ a closed subgroup of a Lie group $G$, we have the following exact sequence of adjoint representations of $H$:
\begin{equation}
	\shortexact*{{\mathfrak{h} }}{\mathfrak{g}}{\mathfrak{g/h}}.
\end{equation}
The canonical principal $H$-bundle $H\rightarrow G/H$ gives an exact functor from representations of $H$ to vector bundles over $G/H$. This gives a corresponding sequence of vector bundles:
\begin{equation}
	\shortexact*{G\times_{H}{\mathfrak{h} }}{G\times_{H}\mathfrak{g}}{G\times_{H}\mathfrak{g/h}}.
\end{equation}
There is an isomorphism $G\times_{H}\mathfrak{g/h} \rightarrow T(G/H)$~\cite[\S 7.4]{BH58}. Let $p\colon G \to G/H$ and $L_X$ be the left invariant vector field generated by $X\in \mathfrak{h}$.
Then the mapping of $(g, X+\mathfrak{h})\in G \times (\mathfrak{g/h})$ to $T_g p \cdot L_X(g) \in T_{gH}(G/H)$ gives the isomorphism. 
Specifically for our problem, we have the $\SO_3$-bundle $\SU_3 \rightarrow W$, which by the present construction gives the desired result.
\end{proof}
\begin{rem}
This is an example of the ``mixing construction": for a principal $G$-bundle $P \rightarrow M$ and a $G$-representation $V$, the space $P \times_G V$
is a vector bundle over $M$ with rank equal to the dimension of $V$.
\end{rem}
Back to $\eta_{W,R}(\mathcal{D}_{\text{Dirac}\times TW})$ in $\eta_{\text{gravitino}}$. In terms of 
$\eta_{R}(\mathcal{D}_{\text{Dirac}})$
this term in~\eqref{eq:gravitinio} unpacks to
\begin{equation}
\eta_{R}(\mathcal{D}_{\text{Dirac}\times TW}) =
    \eta_{W, TW\otimes \boldsymbol 8}(\mathcal D_{\mathrm{Dirac}})\,.
\end{equation}
 We want to show this $\eta$-invariant is for a representation of the form $V_2\oplus V_2$. We therefore need to study the representation $TW\otimes \boldsymbol{8}$. Since $TW$ is the associated bundle to $\su_3/\mathfrak{so}_3$, $\eta_{W, TW\otimes \boldsymbol 8}(\mathcal D_{\mathrm{Dirac}})$ is the $\eta$-invariant for the Dirac operator associated to the $\SO_3$-representation $(\su_3/\mathfrak{so}_3)\otimes\boldsymbol 8$. The $\su_8$-representation $\boldsymbol 8$ factors as $(\boldsymbol 2\oplus \boldsymbol 2)\oplus (\boldsymbol 2\oplus \boldsymbol 2)$ as an $\su_2$-representation by~\eqref{eq:branching}, so we can set $V_2\coloneqq (\su_3/\mathfrak{so}_3)\otimes (\boldsymbol 2\oplus \boldsymbol 2)$ and conclude.
\end{proof}

We now move onto the nonperturbative anomaly from the vector bosons, which is accessible from $\Omega^{\SO}_{5}(B\SO_3)$. By applying \eqref{eq:SObordism} to the AHSS, we only need to consider $H_{5}(B\SO_3,\bZ)$ as well as the $\bZ/2$ element in bidegree $(0,5)$. 
One can evaluate the torsion part of $H_{5}(B\SO_3;\bZ)$ by the universal coefficient theorem; this will involve looking at $H^6(B\SO_3;\bZ)$. We find that this is given by $w_2 w_3$ of the $\SO_3$ bundle and is nontrivial on the Wu manifold. 
Then the AHSS says $\Omega_5^\SO(B\SO_3) = \Z/2\times\Z/2$ detected by the bordism invariants $\int w_2(TM)w_3(TM)$ and $\int w_2(P)w_3(P)$; these are generated by $W$ with the trivial bundle, and $W$ with the principal $\SO_3$-bundle.
We see that while the bosonic anomaly is in principle $\bZ/2 \times \bZ/2$ valued, and coupling to spin structure eliminates one of the $\bZ/2$. Using \eqref{eq:vectorbranching}, for the representation of the vector boson, the anomaly is also twice of something as a bordism invariant. This is reasonable since the anomaly of multiple particles is the tensor product of their anomalies
\footnote{For the gaugino and gravitino we could employ the decomposition of representations directly to the $\eta$-invariant. 
In the case of the vector boson, we use the fact that direct sums of representations goes to tensor products of anomalies.}.
The anomaly for the vector bosons is 2 times something as a bordism invariant, since the perturbative part vanished, and considering that we have argued that everything else in \eqref{eq:totalanomaly} vanishes aside from $\alpha^{\text{np}}_{1}$,  we have that $\alpha = \alpha^{\text{np}}_{1}$. But $\alpha$ is $\bZ/2$ valued, and with $\alpha^{\text{np}}_{1}$ equating to 0 mod 2, the full anomaly vanishes, thus establishing \cref{prop:anomalyVanish}.

\bibliographystyle{alpha}
\bibliography{references}

\end{document}